\documentclass[journal,onecolumn]{IEEEtran}
\usepackage{graphicx}
\usepackage{mathrsfs}
\usepackage{amsmath}
\usepackage[utf8]{inputenc}
\usepackage[english]{babel}
\usepackage{algorithm}

\usepackage{amsthm}
\usepackage{dsfont}
\usepackage{amssymb}
\usepackage[usenames,dvipsnames]{color}
\usepackage{romannum}

\usepackage{algpseudocode}
\usepackage{romannum}

\newtheorem{theorem}{Theorem}
\newtheorem{lemma}{Lemma}
\newtheorem{definition}{Definition}
\newtheorem{remark}{Remark}

\newtheorem{example}{Example}
\newtheorem*{example*}{Example 1 revisited}
\newtheorem*{example3*}{Example 3 revisited}

\newtheorem{corollary}{Corollary}

\begin{document}
%
\title{Cooperative Data Exchange based on MDS Codes}
%

\author{Su Li and Michael Gastpar
\thanks{This paper was partially presented at the 2017 IEEE International Symposium on Information Theory, June, and at the 2017 55th Annual Allerton Conference on Communication, Control, and Computing, October.}}
\markboth{}%
{Shell \MakeLowercase{\textit{et al.}}: Bare Demo of IEEEtran.cls for IEEE Journals}
%



\pagenumbering{arabic}

\maketitle

\begin{abstract}
The cooperative data exchange problem is studied for the fully connected network. 
In this problem, each node initially only possesses a subset of the $K$ packets making up the file. 
Nodes make broadcast transmissions that are received by all other nodes.
The goal is for each node to recover the full file.
In this paper, we present a polynomial-time deterministic algorithm to compute the optimal (i.e., minimal) number of required broadcast transmissions and to determine the precise transmissions to be made by the nodes.
A particular feature of our approach is that {\it each}  of the $K-d$ transmissions is a linear combination of {\it exactly}  $d+1$ packets, and we show how to optimally choose the value of $d.$
We also show how the coefficients of these linear combinations can be chosen by leveraging a connection to Maximum Distance Separable (MDS) codes.
Moreover, we show that our method can be used to solve cooperative data exchange problems with weighted cost as well as the so-called successive local omniscience problem.
\end{abstract}

\begin{IEEEkeywords}
Cooperative data exchange, maximum distance separable codes, linear codes
\end{IEEEkeywords}

%
\IEEEpeerreviewmaketitle

\section{Introduction}
%
%
%
%

Consider a fully connected network composed of $N$ nodes that all want to recover a $K$ packet file. 
Each node initially only has a subset of the packets. 
Each node can generate coded packets by using its locally available packets
and transmit them to other nodes through a lossless broadcast channel, i.e. all other nodes receive the coded packets. 
The goal is for each node to assemble the full file.
The key questions are: (1) What is the minimum number of required transmissions? 
(2) What should individual nodes transmit?
This problem was introduced by El Rouayheb {\it et al.} in~\cite{el2010coding} and is referred to as \textit{Cooperative Data Exchange} (CDE) or \textit{communication for omniscience} for the fully connected network.
Concerning the minimum number of required transmissions, upper and lower bounds were established in~\cite{el2010coding}.
A deterministic algorithm was proposed to produce a coding scheme which achieves universal recovery
using at most twice the minimum number of required transmissions.
The CDE problem can be formulated as an Integer Linear Program (ILP) with the Slepian-Wolf constraints on all proper subsets of the nodes' available packet information.
A randomized algorithm \cite{sprintson2010randomized} and a deterministic algorithm \cite{sprintson2010deterministic} were proposed to give an approximate solution and solve the problem with high probability.
We note that the number of constraints in the ILP at hand grows exponentially with the problem size.
Nevertheless, exact polynomial-time algorithms were found in~\cite{milosavljevic2016efficient,courtade2014coded,8063422} based on minimizing submodular functions and subgradient optimization.

\subsection{Main Contributions}

In this paper, we consider the CDE problem for the fully connected network in a new perspective. Our main contributions can be summarized as follows:
\begin{itemize}
	\item[(1)] We present a new deterministic algorithm to compute the minimal number of required transmissions. It is based on searching for the existence of certain conditional bases of the packet distribution matrix.
	The complexity is bounded by $\mathcal{O}(N^3K^3\log (K))$, significantly lower than the complexity of the best known existing algorithms proposed in~\cite{milosavljevic2016efficient} based on minimizing submodular functions $\mathcal{O}((N^6K^3+N^7)\log(K))$ and based on subgradient methods $\mathcal{O}((N^4\log(N)+N^4K^3)K^2\log(K))$.
	\item[(2)] We propose a novel coding scheme with $K-d$ transmissions in which each transmission is a linear combination of $d+1$ packets for any $0 \le d < K$.
	Nodes with at least $d$ packets can recover their missing packets from this coding scheme regardless of which packets they have in detail. 
	The coefficient matrix of the coding scheme can be efficiently generated by performing elementary row operations on Vandermonde matrices.
	\item[(3)] We show that the CDE problem with weighted cost objective function and successive local omniscience problem can be solved by our method with slightly modifications. The complexity of our method for solving CDE problem with weighted cost objective function is bounded by $\mathcal{O}(N^3K^3\log (K))$, which is the same as CDE problem without weighted cost. For successive local omniscience problem with $M$ priority groups, our methods has complexity bounded by $\mathcal{O}(N^3K^3M\log (K))$. For both generalized problem, the way of constructing coding scheme is the same as what we do for basic CDE problem.
\end{itemize}

\subsection{Further Related Work}

The CDE problem was extended to general network topologies, and it was shown that linear codes are sufficient to optimally solve the CDE problem in~\cite{courtade2014coded,gonen2015coded}. 
However, the same work also revealed that for arbitrarily connected networks, the CDE problem is NP-hard and cannot be solved exactly with polynomial time algorithms. 
Many extensions of the CDE problem have also been studied.
In~\cite{7541714}, the nodes are divided into two classes, high and low priority. The resulting CDE problem with priorities was formulated as a multi-objective integer linear program. Assuming a uniformly random packet distribution and restricting to the limit as the number of packets tends to infinity, a closed-form expression for the minimal number of required transmissions was derived.
In~\cite{milosavljevic2016efficient,5743607}, transmissions sent by different node are considered to have different cost. Instead of minimizing the total number of transmissions, the goal becomes minimizing the total cost, {\it i.e.,} a weighted sum of the transmissions. 
To solve the CDE problem with weighted cost, a deterministic polynomial algorithm based on submodular function minimization was proposed in~\cite{milosavljevic2016efficient}, while a randomized greedy algorithm was proposed in~\cite{5743607}.     
In~\cite{courtade2014coded,tajbakhsh2011generalized}, it is assumed that each packet can be split into the same number of smaller chunks and the optimization goal is minimizing the normalized total number of transmissions. 
Intuitively, the larger the number of chunks we split each packet into, the smaller the normalized total number of transmissions that can be achieved, and it has been proved that it is sufficient to split each packet into $N-1$ chunks.
In~\cite{7447045}, the nodes are divided into two classes, reliable and unreliable. For unreliable nodes, the initially available packets are unknown (but it is known how many packets they have) and the packet transmissions are subject to arbitrary erasures. A closed-form expression for the minimal number of transmissions for the case of only a single unreliable node was derived with probability approaching 1 as the number of packets tends to infinity. For more than one unreliable node, an approximate solution was provided.

The CDE problem for the fully connected network is also related to the secret key generation problem, which was introduced in~\cite{1362897} and was formulated as a maximization problem over all partitions of the node set.
Tyagi {\it et al.}~\cite{7902138} leveraged this to derive an algorithm which achieves local omniscience in each step and outputs a sub-optimal solution.
The weakly secure data exchange problem was introduced in~\cite{yan2013algorithms}. The goal is to achieve universal recovery while revealing as little information as possible.
In contrast to the coding scheme in~\cite{yan2013algorithms} in which each transmission is a linear combination of as many packet as possible, our scheme considers a fixed number of packets for every transmission. 
In general cases, the communication rate in the secret key generation problem is asymptotic which makes the problem NP-hard, 
while the minimal number of required transmissions in the CDE problem is integer so that it can be solved by polynomial-time algorithms (as previously mentioned).

For solving the CCDE problem, only knowing the minimal number of required transmissions is not enough. 
It is also necessary to design the coding scheme. 
Given the total number of transmissions, designing the coding scheme is a multicast network code construction problem 
and can be solved by the polynomial time algorithm proposed by Jaggi {\it et al.} in~\cite{jaggi2005polynomial}.

\subsection{Organization}
This paper is organized as follows. Section~\ref{sec:systemmodel} formally defines our system model and introduces definitions and notations that would be used in this paper. Section~\ref{sec:mainresults} presents our main results. Section~\ref{sec:algorithm} proposes our algorithms to compute the minimal number of required transmissions for the basic CDE problem. 
Section~\ref{sec:codeconstruction} presents an efficient way to construct the linear code. 
Section~\ref{sec:weightcost} and Section~\ref{sec:slo} show our method can be used to efficiently solve two generalized topic, CDE with weighted cost and successive local omniscience.
Section~\ref{sec:conclusion} concludes our work.

\section{System Model and Definitions}\label{sec:systemmodel}
Before we formally introduce the problem, we define some notations.
Let $[n]$ denote the integer set $\{1,\dots, n\}$.
For any vector $u$, we use $u_i$ to denote the $i^{th}$ entry of $u$.
For any matrix $E$, we use $E_{ij}$ to denote the entry at $i^{th}$ row $j^{th}$ column.
Let $w_H(u)$ denote the number of non-zero entries in vector $u$.
For any set of vectors $\mathbf{U} = \{u_1,u_2,\dots\}$ and subset of vectors $\mathbf{S} \subseteq \mathbf{U},$ let
$u_{\mathbf{S}}$ denote the bitwise \textbf{OR} (or, equivalently, the componentwise maximum) of all the vectors in the set $\mathbf{S}$.

Consider a fully connected network which has $N$ nodes and a desired file composed of $K$ packets. 
Let $\mathbf{N} = [N]$ and $\mathbf{P} = \{P_i,i\in[K]\}$ denote the set of nodes and set of packets, respectively.
Each $P_i \in \mathbb{F}$, where $\mathbb{F}$ is some finite field.
Without loss of generality, we assume that every packet is initially available at least at two nodes and at most at $N-1$ nodes\footnote{If there is a packet that is only initially available at one node, the optimal strategy is just letting that node send the uncoded packet to the others. If there is a packet that is available at all nodes, then no one needs to recover it.}.
The set of the packets initially available at node $i$ is denoted as $\mathbf{X}_i$ ($\forall i \in \mathbf{N}:  \mathbf{X}_i\subseteq \mathbf{P}$).
The union set of the packets initially available at a subset of node $\mathbf{I}\subseteq \mathbf{N}$ is denoted as $\mathbf{X}_\mathbf{I} = \bigcup_{i\in \mathbf{I}} \mathbf{X}_i$.
We assume that all the nodes collectively have all packets, which means $\mathbf{X}_\mathbf{N} = \mathbf{P}$.
The notation $\mathbf{X}_\mathbf{I}^c = \mathbf{P} \setminus \mathbf{X}_\mathbf{I}$ denotes the jointly missing packets at nodes in set $\mathbf{I}$.
Let $\mathcal{M} = \min_{i\in \mathbf{N}} |\mathbf{X}_i|$ be the minimum number of initially available packets at any single node.

\begin{definition}
	Define the packet distribution matrix $E$ as the $N\times K$ matrix with entry at $i^{th}$ row $j^{th}$ column:
		\begin{align}
			E_{ij} =
			\left\{
			\begin{aligned}
				1&, &P_j \in \mathbf{X}_i\\
				0&, &otherwise 
			\end{aligned}
			\right. 
		\end{align}
	We will refer to the $K$-dimensional binary (row) vector $e_i$, the $i^{th}$ row of $E$, as the Packet Distribution Vector (PDV) of node $i$. 
\end{definition}

Let $\mathbf{T} = \{T_1,\dots,T_R\}$ denote a linear coding scheme with $R$ transmissions\footnote{Only linear coding schemes are considered since it has been proved that they are sufficient to optimally solve the CDE problem~\cite{courtade2014coded,gonen2015coded}.},
which means that each transmission $T_i$ is a linear combination of packets available at the sender node.
Let $\mathbf{r} = [r_1,\dots,r_N]^\mathsf{T}$ denote the rate vector where each $r_i$ is the number of transmissions made by node $i$.
Hence, the total number of transmissions can be expressed as $R = \sum_{i=1}^{N} r_i$.
Let $R^*$ denote the minimal number of required transmissions.
Define the coefficient matrix as matrix $A$ with entries $a_{ij} (i\in [R], j\in [K]$), where $\alpha_i = [a_{i1},\dots,a_{iK}]$ and $\beta_j = [a_{1j},\dots,a_{Rj}]^\mathsf{T}$ are the $i^{th}$ row and $j^{th}$ column vectors of $A$, respectively. 
Then we have:
\begin{align}
\begin{bmatrix}
T_1       \\
T_2\\
\vdots\\
T_R
\end{bmatrix}
=
\begin{bmatrix}
a_{11} & a_{12} & \dots  & a_{1K} \\
a_{21} & a_{22} & \dots  & a_{2K} \\
\vdots & \vdots & \ddots & \vdots \\
a_{R1} & a_{R2} & \dots  & a_{RK}
\end{bmatrix}
\begin{bmatrix}
P_1       \\
P_2\\
\vdots\\
P_K
\end{bmatrix}  \label{eq-Amatrix}
=
\begin{bmatrix}
\alpha_1\\\alpha_2\\\vdots\\\alpha_R
\end{bmatrix}
\begin{bmatrix}
P_1       \\
P_2\\
\vdots\\
P_K
\end{bmatrix}
=
\begin{bmatrix}
\beta_1&\beta_2&\dots&\beta_K
\end{bmatrix}
\begin{bmatrix}
P_1       \\
P_2\\
\vdots\\
P_K
\end{bmatrix}
\end{align}

It has been shown that any rate vector $\mathbf{r}$ which achieves universal recovery should satisfy the following Slepian-Wolf constraints~\cite{courtade2010optimal}:
\begin{align}
\sum_{i\in\mathbf{N}\setminus \mathbf{I}}r_i \ge \left| \mathbf{X}_\mathbf{I}^c\right|, \forall \mathbf{I}\subsetneq \mathbf{N}\label{eq:sw-constr}
\end{align}
Let $\Omega = \{\mathbf{r} = [r_1,\dots,r_N]^\mathsf{T}: \sum_{i\in\mathbf{N}\setminus \mathbf{I}}r_i \ge \left| \mathbf{X}_\mathbf{I}^c\right|, \forall \mathbf{I}\subsetneq \mathbf{N}\}$ denote the set of all rate vectors $\mathbf{r}$ which satisfy~\eqref{eq:sw-constr}.
The minimal number of required transmissions for achieving universal recovery can be computed by solving the following integer linear program:
\begin{align}
R^*  = \min_{\mathbf{r} \in \Omega} &\sum_{i=1}^{N} r_i  . \label{ILP}
\end{align}

\begin{example}
	\label{EX:CCDE}
	Consider a CDE problem for the fully connected network with $N=4$ nodes and $K=9$ packets. 
	The packet distribution matrix is as follows:
	
	\begin{equation*}
	E =\begin{bmatrix}
	1  & 1 & 1 & 1   & 1 & 1  & 0  & 0 & 0  \\
	1  & 1 & 1 & 0 & 0  & 0  & 1 & 1  & 1 \\
	0  & 0 & 0 & 1 & 1  & 1  & 1 & 1  & 1 \\
	1 & 0 & 1 & 0& 0 & 1 & 0 & 1 &0
	\end{bmatrix}
	\end{equation*}	
	
	The number of non-empty proper subset of nodes is 18.
	Thus, we can write down 18 linear constraints and solve the inequalities. For example, for $\mathbf{I} = \{1\}$, the constraint for total number of transmissions made by nodes $\{2,3,4\}$ is
	\begin{align}
		\sum_{i=\{2,3,4\}} r_i \ge |\mathbf{X}_{1}^c| = 3
	\end{align}
	By using the methods proposed in~\cite{milosavljevic2016efficient,courtade2014coded}, based on minimizing submodular function,
	the integer linear program can be solved in polynomial time and the minimal number of required transmissions should be $5$.
	After knowing the minimal number of transmissions, 
	generating the coding scheme is a multicast network code construction problem 
	and can be solved by polynomial time algorithms proposed in~\cite{jaggi2005polynomial}.
	One feasible coding scheme could be: node 1 sends $T_1 = P_1+P_5$ and $T_2 = P_2+P_6$, node 2 sends $T_3 = P_3+P_7$, node 3 sends $T_4 = P_4+P_8$ and $T_5 = P_9$.
\end{example}

In general, there are multiple different optimal coding schemes that achieve universal recovery.
Although not all nodes have to make transmissions, 
the existing algorithms which solve the integer linear program~\eqref{ILP} have to consider constraints introduced by all non-empty proper subset of nodes. 
In Example~\ref{EX:CCDE}, the optimal coding scheme does not require node $4$ to make any transmission, 
but the algorithms still have to consider the constraints related to node $4$.
However, we will show that without knowing the exact packet distribution information at some nodes (in this example, node 4), but only knowing the number of initially available packets at them, it is still possible to compute the minimum number of required transmissions and construct the optimal  coding scheme which achieves universal recovery with the smallest number of transmissions. 

\begin{definition}[$(d,K)$-Basis]
	A set of $K$-dimensional binary linearly independent vectors $(\mathbf{V} = \{v_i: i \in [K-d]\},0\le d\le K-1)$ is called a $(d,K)$-Basis if
	\begin{align}
	&w_H(v_{\mathbf{S}}) \ge |\mathbf{S}|+d,&\ \forall \emptyset \ne \mathbf{S}\subseteq \mathbf{V} . \label{Eq:dB2}
	\end{align}
\end{definition}

\begin{definition}[Balanced $(d,K)$-Basis]
	A $(d,K)$-Basis $(\mathbf{V} = \{v_i: i \in [K-d]\},0\le d\le K-1)$ is called a balanced $(d,K)$-Basis if
	\begin{align}
		&w_H(v_i) = d+1,&\ \forall i\in [K-d] .\label{Eq:dB1}
	\end{align}
\end{definition}

Condition~(\ref{Eq:dB2}) requires that $w_H(v_\mathbf{S})$, the number of dimensions spanned by vectors in $\mathbf{S}$, be no less than the number of vectors plus $d$. Hence, the number of vectors in each subspace of the $K$-dimensional space is limited.

\begin{definition}
	A binary vector $u$ can generate another binary vector $v$ if $u$ and $v$ have the same dimension and 
	\begin{equation}
	\{m: v_{m} = 1 \} \subseteq \{n: u_{n} = 1 \}.
	\end{equation}
	Moreover, let $\mathcal{G}(u)$ denote the set of all binary vectors that can be generated by $u$. Define $\mathcal{G}(\mathbf{S}) = \cup_{u\in \mathbf{S}} \mathcal{G}(u)$ and $\mathcal{G}(u,d) = \{v: v\in \mathcal{G}(u), w_H(v)=d+1\}$. 
\end{definition}


\begin{definition}
	A set of $K$-dimensional binary vectors $\mathbf{U} =\{u_1,\dots,u_L\}$ is able to generate a $(d,K)$-Basis $\{v_i: i\in [K-d]\}$ if $\forall i\in [K-d]$, $v_i \in \mathcal{G}(\mathbf{U},d)$.
	Let $(d^*,K)$-Basis denote the $(d,K)$-Basis with largest $d$ that can be generated by given vectors.
\end{definition}


\begin{lemma}\label{LM:d}
	If a set of $K$-dimensional binary vectors is able to generate a $(d_1,K)$-Basis, then it is also able to generate a $(d_2,K)$-Basis for any $d_2\le d_1$.
\end{lemma}

\begin{proof}
	Consider a set of binary vectors $\{u_1,\dots,u_N\}$ that is able to generate a $(d_1,K)$-Basis $\mathbf{V} = \{v_1,\dots,v_{K-d_1}\}.$ Then
	\begin{align}
		\forall i \in [K-d_1], \exists j \in [N]: \{m: v_{im}\}\subseteq \{n:u_{jn}\} 
	\end{align}
	Hence any vector generated by $v_i$ should also be able to be generated by the corresponding $u_j$.
	Thus, to prove this lemma, it suffices to show that $\forall d_2 \le d_1$, there exists a $(d_2,K)$-Basis $\mathbf{Q} =\{q_1,\dots,q_{K-d_2}\}$ that can be generated by $\{v_1,\dots,v_{K-d_1}\}$.
	Since $\mathbf{V}$ is a $(d_1,K)$-Basis and $d_2\le d_1$, $\forall \mathbf{S} \subseteq\mathbf{V}$, we have
	\begin{align}
		w_H(v_{\mathbf{S}}) \ge |\mathbf{S}| + d_1 \ge |\mathbf{S}| + d_2
	\end{align}
	Thus all vectors in $\{v_1,\dots,v_{K-d_1}\}$ satisfy the constraints for vectors of $(d_2,K)$-Basis.
	We can choose $q_i = v_i$,  $\forall i = [K-d_1]$.
	Moreover, $\forall j \in \{K-d_1+1,\dots,K-d_2\}$, we choose $q_{j} = v_1$ to be the repeated vector.
	Then, $\forall \hat{\mathbf{S}} \subseteq \mathbf{Q}$:
	\begin{align}
		w_H(q_{\hat{\mathbf{S}}}) \ge |\hat{\mathbf{S}}| +d_1 - c \ge |\hat{\mathbf{S}}|+d_2
	\end{align} 
	where $c = |\hat{\mathbf{S}}\cap\{q_j: j\in \{K-d_1+1,\dots,K-d_2\}\}| \le d_1-d_2$ is the number of the repeated vectors.
	Hence $\mathbf{Q} = \{q_1,\dots,q_{K-d_2}\}$ is a $(d_2,K)$-Basis.  
\end{proof}

\section{Main Results}\label{sec:mainresults}
In this section, we present our main results and proofs.
The relationship between a $(d,K)$-Basis and a coding scheme that can enable nodes with at least $d$ packets to recover all missing packets is revealed by the following theorem.

\begin{theorem}\label{Thm:dt}
	If for some subset of nodes $\mathbf{I} \subseteq \mathbf{N}$ there exists a $(d,K)$-Basis $\mathbf{V} \subseteq \mathcal{G}(\{e_i,i\in \mathbf{I}\},d)$, then the nodes of $\mathbf{I}$ can generate a coding scheme $\mathbf{T} = \{T_1,\dots,T_R\}$ with $R= K-d$ such that $\forall i \in \mathbf{N}, w_H(e_i)\ge d$, node $i$ can recover all packets. 
\end{theorem}

\begin{proof}
	In our coding scheme, each transmission $T_i$ is a linear combination (with appropriate coefficients) of the packets indexed by the non-zero entries in $v_i$.
	Since the vectors $v_i$'s are a subset of the vectors generated by the PDVs of the nodes in $\mathbf{I},$ there is one node in $\mathbf{I}$
	for each $v_i$ that can locally produce and transmit said linear combination.
	The overall code can thus be characterized by a matrix $A$ as in Eqn.~\eqref{eq-Amatrix} where in row $i,$ only the elements indexed by $v_i$ are non-zero.
	
	For any $\mathcal{C} \subset [K]$ with $|\mathcal{C}| = R,$ let $A(\mathcal{C})$ denote the submatrix of $A$ consisting of the $R$ columns indexed by $\mathcal{C}.$
	Due to constraint~(\ref{Eq:dB2}), $\forall \emptyset \ne \mathbf{S} \subseteq V$, we have
	\begin{align}
		w_H(v_{\mathbf{S}}) \ge |\mathbf{S}|+d.
	\end{align} 
	Denote the $i^{th}$ of row of $A(\mathcal{C})$ by $\alpha_i(\mathcal{C})$.
	Then $\forall \hat{\mathbf{S}} \subseteq  \{\alpha_1(\mathcal{C}), \cdots, \alpha_R(\mathcal{C})\} $, we have
	\begin{align}
		w_H(\alpha_{\hat{\mathbf{S}}}(\mathcal{C})) \ge w_H(v_{\hat{\mathbf{S}}}) - d  \ge |\hat{\mathbf{S}}|.\label{Eq:hm}
	\end{align} 
	Let $\mathbb{G}(A(\mathcal{C}))$ denote the bipartite graph corresponding to $A(\mathcal{C})$, where there is an edge between $i^{th}$ left vertex and $j^{th}$ right vertex if and only if $A(\mathcal{C})_{ij}\ne 0$.
	Since Enq.~\eqref{Eq:hm} satisfies the condition of Hall's marriage theorem, there exists a perfect matching in $\mathbb{G}(A(\mathcal{C}))$.
	According to Edmond's Theorem~\cite{motwani2010randomized}, the existence of perfect matching in bipartite graph $\mathbb{G}(A(\mathcal{C}))$ implies that $\det(A(\mathcal{C})) \not \equiv 0$.
	
	The product of determinants of all submatrices with $R$ columns, denoted by $\prod_{\mathcal{C}}\det(A(\mathcal{C}))$, is a multivariate polynomial of non-zeros entries of $A$.
	For a large enough finite field, there always exists a good choice of non-zero entries of $A$ such that  $\prod_{\mathcal{C}}\det(A(\mathcal{C})) \not \equiv 0$~\cite{ho2006random}.
	For such choices, any $R$ columns of $A$ can be linearly independent at the same time.
	In other words, given any $d$ packets, the other $R$ missing packets can be recovered from our coding scheme.
\end{proof}

\begin{remark}
	In Theorem~\ref{Thm:dt}, we proved that if the PDVs of nodes are able to generate a $(d,K)$-Basis, they can also generate a coding scheme such that nodes with at least $d$ packets can recover all missing packets from the coding scheme. 
	The coefficient matrix used in the proposed coding scheme can be associated with a constrained generator matrix for an MDS code~\cite{dau2013balanced}. We will introduce an efficient way to construct it by performing elementary row operations on a Vandermonde matrix in Section~\ref{sec:codeconstruction}.
\end{remark}

Theorem~\ref{Thm:dt} characterizes a certain class of coding schemes. Their unique feature is that each transmission is a (judiciously chosen) linear combination of exactly the same number of pure packets, namely, $d+1.$ Initially, this last feature may appear to be too restrictive to attain optimal performance. However, in the sequel, we will establish in two steps that there always exists an optimal scheme with this special property. Nonetheless, let us recall that in general, the optimal data exchange scheme is not unique, so there may be alternative schemes attaining the same (optimal) number of transmissions while not satisfying the special property.
To establish existence of an optimal scheme with the special property, we will next establish that if a (linear) scheme enabling universal recovery exists, then the nodes are also able to generate a corresponding basis (and hence, by Theorem~\ref{Thm:dt}, a scheme with the special property must exist). More precisely, we have the following theorem:


\begin{theorem}\label{Thm:2}
	If a subset of nodes is able to generate a linear coding scheme with $R$ ($R=K-d$) transmissions which achieves universal recovery,
	then the PDVs of the nodes can generate a $(d,K)$-Basis $\mathbf{V} = \{v_1\dots,v_{R}\}$.
\end{theorem}

\begin{proof}
	We assume that a subset of nodes $\mathbf{I}$ can generate $R$ linearly independent transmissions $\hat{\mathbf{T}} = \{\hat{T}_1,\dots,\hat{T}_{R}\}$ which achieves universal recovery.
	The code can be charaterized by a matrix $\hat{A}$ as in Eqn.~\eqref{eq-Amatrix} with rows $\hat{\alpha}_i$'s and columns $\hat{\beta}_j$'s. 
	Let $\hat{\mathbf{V}} = \{\hat{v}_1,\dots,\hat{v}_R\}$ where each $v_i = supp(\hat{\alpha}_i)$.
	That means in row $i$ of $\hat{A}$, only the elements indexed by $\hat{v}_i$ are non-zero.
	We would like to show that if $\hat{\mathbf{V}} = \{\hat{v}_1,\dots,\hat{v}_R\}$ does not satisfy Constraint~\eqref{Eq:dB2} of the $(d,K)$-Basis, then the nodes which generate the corresponding transmissions are able to add more packets into the linear combinations until the Constraint~\eqref{Eq:dB2} is satisfied. 
	
	For each non-empty subset $\mathbf{S}\subseteq \{\hat{\alpha}_1,\dots,\hat{\alpha}_R\}$ such that $w_H(\hat{\alpha}_{\mathbf{S}}) < |\mathbf{S}|+d$,
	we have
	\begin{align}
		K-w_H(\hat{\alpha}_{\mathbf{S}}) > K- |\mathbf{S}|-d \ge R-|\mathbf{S}|+1
	\end{align} 
	For the row vectors in $\mathbf{S}$, at least $R-|\mathbf{S}|+1$ columns are all zeros.
	Hence, there must exist a subset of columns $\mathcal{C}\subset [K]$ and corresponding subset of column vectors $\mathbf{C}\subseteq \{\hat{\beta}_1,\dots,\hat{\beta}_K\}$ such that
	\begin{align}
		|\mathcal{C}|=|\mathbf{C}|=R-|\mathbf{S}|+1
	\end{align}
	\begin{align}
		R-w_H(\hat{\beta}_{\mathbf{C}}) \ge |\mathbf{S}| \Rightarrow w_H(\hat{\beta}_{\mathbf{C}}) \le R- |\mathbf{S}| < |\mathbf{C}|
	\end{align}
	Let $\hat{A}(\mathcal{C})$ denote the submatrix which is composed of the columns indicated by subset of column vectors $\mathcal{C}$.
	Then submatrix $\hat{A}(\mathcal{C})$  is rank deficient.
	Let $P_\mathcal{C}\doteq \{P_i:i\in\mathcal{C}\}$ denote the set of packets indexed by $\mathcal{C}$.
	If the set $\mathcal{N}$ of nodes that generate transmissions $\{\hat{T}_i: \alpha_i\in \mathbf{S}\}$ cannot add any more packets into the linear combination for their transmissions,
	they have no more extra available packets in $P_\mathcal{C}$ and each transmission is a linear combination of all its sender node's available packets.
	This assumption leads to a contradiction that nodes in $\mathcal{N}$ cannot recover all missing packets. 
	Thus, nodes in $\mathcal{N}$ must have more packets in $P_\mathcal{C}$ and can add them into the linear combination to generate new transmissions $\{T_i:\alpha_i\in\mathbf{S}\}$ such that $w_H(\alpha_{\mathbf{S}}) = |\mathbf{S}|+d$, where
	$\alpha_i$ denotes the coefficient vector of transmission $T_i$.
	By replacing $\{\hat{T}_i:\alpha_i \in \mathbf{S}\}$ with $\{T_i:\alpha_i\in \mathbf{S}\}$, we have a new coding scheme $\mathbf{T}$ such that the set of corresponding support vectors $\mathbf{V}=\{v_i,\dots,v_R\}$ forms a $(d,K)$-Basis.
	For each transmission $T_i$ and the corresponding $\hat{T}_i$, we have $\hat{v_i}\in \mathcal{G}(v_i)$.
	Given that $\hat{\mathbf{T}}$ can achieve universal recovery, $\mathbf{T}$ can also achieve universal recovery.
\end{proof}

\begin{lemma}\label{LM:balanceddk}
	If a subset of nodes can generate a linear coding scheme based on $(d,K)$-Basis which enables nodes with at least $d$ packets to recover all packets, they also can generate an equivalent linear coding scheme based on balanced $(d,K)$-Basis.
\end{lemma}
\begin{proof}
	For any linear coding scheme $\mathbf{T} = \{T_1,\dots,T_{K-d}\}$ based on $(d,K)$-Basis $\mathbf{V}=\{v_1,\dots,v_{K-d}\},$
	let $A$ denote the coefficient matrix of $\mathbf{T}$ and $\alpha_i$ denote the $i^{th}$ row of $A$.
	For each $T_i$ with $w_H(\alpha_i) > d+1$, we show that it can be reduced to a linear combination of $d+1$ packets.	
	$\forall \tilde{\mathbf{S}} \subseteq \{\alpha_j: j\ne i\}$, $w_H(\alpha_{\tilde{\mathbf{S}}})\ge |\tilde{\mathbf{S}}|+d$.
	The linear combination of $\{T_j: j\ne i\}$ can provide $K-d-1$ degree of freedoms among the used packets.
	Hence, by subtracting a proper linear combination of $\{T_j: j\ne i\}$ from $T_i$, we can get $\bar{T_i}$ with $w_H(\bar{\alpha}_i) = d+1$.
	Thus the corresponding $\bar{\mathbf{V}}$ is a balanced $(d,K)$-Basis.
\end{proof}

\begin{example*}
	We already know a coding scheme with $5$ transmissions that achieves universal recovery. 
	But each coded packet for transmission is a linear combination of two packets or just one pure packet. 
	According to Theorem~\ref{Thm:2} and Lemma~\ref{LM:balanceddk}, there must exist another coding scheme in which every coded packet for transmission is a linear combination of $5$ packets. It is easy to verify that coding scheme with the following coefficient matrix (over finite field $GF(2^4)$ with primitive polynomial $\alpha^4+\alpha+1$) also achieves universal recovery.
	\begin{equation}\label{eq:A}
	A =
	\begin{bmatrix}
	5 & 4 & 4 & 1 &1&0&0&0  &0 \\
	15 & 11 & 14 & 14 &0&1&0&0 & 0 \\
	3 & 6 & 13 &0&0&0&15&14 & 0 \\
	9 & 12 & 7 &0&0&0&15&0 & 14 \\
	0 & 0&0  & 10& 14 &6&9&8&0 
	\end{bmatrix}
	\end{equation}
	Each transmission is a linear combination of $5$ packets.
	Define binary matrix $V$ such that 
		\begin{align}
			V_{ij} =
			\left\{
			\begin{aligned}
			&1, &A_{ij} \ne 0\\
			&0, &A_{ij} = 0
			\end{aligned}
			\right.
		\end{align}
	Then we have
	\begin{equation}\label{eq:V}
	V=
	\begin{bmatrix}
	v_1\\v_2\\v_3\\v_4\\v_5
	\end{bmatrix}
	=
	\begin{bmatrix}
	1 & 1 & 1 & 1 &1&0&0&0  &0 \\
	1 & 1 & 1 & 1 &0&1&0&0 & 0 \\
	1 & 1 & 1 &0&0&0&1&1 & 0 \\
	1 & 1 & 1 &0&0&0&1&0 & 1 \\
	0 & 0&0  & 1& 1 &1&1&1&0 
	\end{bmatrix}
	\end{equation}
	The row vectors of $V$ actually form a balanced $(4,9)$-Basis.
	As mentioned in Theorem~\ref{Thm:dt}, given any 4 packets, the other 5 packets can be recovered from the coding scheme based on coefficient matrix $A$. 
	Hence, in this example, the detail information of available packets at node 4 is not necessary. 
	As long as it initially has 4 packets, it can always recover the other packets by receiving these coded packets. 
\end{example*}

Now we have the connection between optimal coding schemes with minimum number of required transmissions and  balanced $(d,K)$-Bases. 
Thus, we can search balanced $(d,K)$-Bases to get achievable (upper) bounds on the minimum number of required transmissions.
Extending the search over all values of $d$ (and using Theorem~\ref{Thm:2} and Lemma~\ref{LM:balanceddk}) then establishes optimal performance.
More precisely, we have the following theorem:

\begin{theorem}\label{Thm:CCDE}
	For the CDE in the fully connected network, the minimal number of required transmissions $R^*$ satisfies:
	\begin{equation}
	R^* = K-\min \{\mathcal{M}, d^*\}
	\end{equation}
	where the $(d^*,K)$-Basis is the largest $(d,K)$-Basis that can be generated by the PDVs of nodes.
\end{theorem}
\begin{proof}
	%

        By assumption, $d^*$ is the largest value of $d$ for which a $(d,K)$-Basis can be generated by the PDVs. But then, by Theorem~\ref{Thm:2}, there does not exist any linear coding scheme that can achieve universal recovery by using fewer than $K-d^*$ transmissions.  
	
	Suppose that $\mathcal{M} \ge d^*$. Then every node has at least $d^*$ packets. Since a $(d^*,K)$-Basis can be generated by the PDVs, according to Theorem~\ref{Thm:dt}, there is a linear coding scheme with $K-d^*$ transmissions such that every nodes with at least $d^*$ packets can recover all missing packets. 
	
	Now suppose that $\mathcal{M} < d^*$. According to Lemma~\ref{LM:d}, the PDVs can also generate a $(d,K)$-Basis with $d= \mathcal{M}$.
	According to Theorem~\ref{Thm:dt}, there is a linear coding scheme with $K-\mathcal{M}$ transmissions such that every nodes with at least $\mathcal{M}$ packets can recover all missing packets. 
	
	Hence, the minimum number of required transmissions satisfies $R^* = K-\min \{\mathcal{M}, d^*\}$.
\end{proof}

\section{Algorithm}\label{sec:algorithm}

According to Theorem~\ref{Thm:CCDE}, to solve the CDE problem for the fully connected network, we need to find the largest value of $d$ such that a $(d,K)$-Basis that can be generated by the PDVs of nodes. We denote this optimal value of $d$ by $d^*.$
This problem can be decomposed into two subproblems:
\begin{itemize}
	\item[(1)] Given a fixed $d$, determine whether any balanced $(d,K)$-Basis can be generated by the PDVs of nodes or not.
	\item[(2)] Find the maximum value of $d$ such that the PDVs of nodes can generate one balanced $(d,K)$-Basis.
\end{itemize}

\subsection{Existence of $(d,K)$-Basis}

Given the packet distribution matrix $E$ and a specific parameter $d$, 
Algorithm~\ref{alg_sdb} is proposed to check whether any balanced $(d,K)$-Basis can be generated by the PDVs of nodes or not.
Due to constraint~\eqref{Eq:dB2}, only nodes with at least $d+1$ packets can generate the $(d,K)$-Basis vectors.
Hence, we only consider the PDVs with $w_H(e_i) > d$ as the candidates to generate basis vectors.

\begin{definition}
	For any binary vector $u$ with $w_H(u) > d$, let $\mathcal{J} = \{j_1,\dots,j_{w_H(u)} \}$ denote the set of indices of the non-zero entries of $u$.
	Define the set $\mathcal{B}(u,d) = \{b_i: i\in [w_H(u)-d]\}$, where $b_{ij_k}$ is the $j_k^{th}$ entry of vector $b_i$ and satisfies 
	 \begin{align}
	 b_{ij_k} = \left\{
	 \begin{aligned}
	 &1, & \text{if } k \in [d] \cup \{d+i\} \text{ and } j_k \in \mathcal{J},\\
	 &0, & \text{otherwise.}
	 \end{aligned}
	 \right.
	 \end{align}
\end{definition}
The set $\mathcal{B}(u,d)$ is a particular set of binary vectors which are generated by $u$. 
Specifically, each of the vectors in the set has weight $d+1$ and it can be verified that the vectors satisfy the Constraint~\eqref{Eq:dB2} of the definition of the $(d,K)$-Basis.
Therefore, they are basis vector candidates for balanced $(d,K)$-Basis.

\begin{example}
	Given the $e_1 = [1,1,1,1,1,1,0,0,0]$, we can assign the $\mathcal{B}(e_1,4) = \{b_1,b_2\}$ where $b_1 = [1,1,1,1,1,0,0,0,0]$ and $b_2 = [1,1,1,1,0,1,0,0,0]$.	
\end{example}

\begin{lemma}
	For any binary vector $v \in \mathcal{G}(u,d) \setminus \mathcal{B}(u,d)$, let $\mathbf{S} = \mathcal{B}(u,d) \cup \{v\}$.
	We have $w_H(b_{\mathbf{S}})< |\mathbf{S}| + d$. 
\end{lemma}
\begin{proof}
	Since $\mathcal{B}(u,d) \subset \mathcal{G}(u,d)$, $w_H(b_{\mathcal{B}(u,d)}) \le w_H(u)$. 
	Also, according to the definition of $\mathcal{B}(u,d)$, we have 
	\begin{align}
		w_H(b_{\mathcal{B}(u,d)}) \ge |\mathcal{B}(u,d)|+d = w_H(u)
	\end{align}
	Hence, $w_H(b_{\mathcal{B}(u,d)}) = w_H(u)$. 
	For any $v \in \mathcal{G}(u,d) \setminus \mathcal{B}(u,d)$, $|\mathbf{S}| = |\mathcal{B}| +1 = w_H(u) -d +1$
	\begin{align}
		w_H(b_{\mathbf{S}}) = w_H(u) < w_H(u)-d+1+d = |\mathbf{S}| + d 
	\end{align} 
\end{proof}
Thus, any vector $v \in \mathcal{G}(u,d) \setminus \mathcal{B}(u,d)$ is not compatible with $\mathcal{B}(u,d)$ in terms of the Constraint~\eqref{Eq:dB2}.
\begin{corollary}
	For each PDV $e_i$,  it is sufficient to check vectors of $\mathcal{B}(e_i,d)$ instead of all vectors of $\mathcal{G}(e_i,d)$.
\end{corollary}
Although for each PDV $e_i$, there are as many as $w_H(e_i) \choose d+1$ balanced $(d,K)$-Basis vectors that can be generated, we can select any $\mathcal{B}(e_i,d)$ and only consider them as the candidate basis vectors.
Any other $v\in \mathcal{G}(e_i,d) \setminus \mathcal{B}(e_i,d)$ can be ignored.

\begin{lemma}
	Let $\mathbf{S} = \{v_1,\dots,v_{|\mathbf{S}|}\}$ denote a set of binary vectors with weight $w_H(v_i) = d+1, \forall v_i\in \mathbf{S}$ and $v_{\mathbf{S}}$ denote the bitwise \textbf{OR} result of all vectors in $\mathbf{S}$. For any vector $v \in \mathcal{G}(v_{\mathbf{S}},d)\setminus \mathbf{S}$, let $\hat{\mathbf{S}} = \mathbf{S}\cup \{v\}$, we have $w_H(v_{\hat{\mathbf{S}}}) < |\hat{\mathbf{S}}| + d$ if 
	\begin{equation}
		w_H(v_{\mathbf{S}}) \le  \sum_{i\in \mathbf{S}} w_H(v_i) - (|\mathbf{S}|-1)d\label{Ieq:Mq1}
	\end{equation}
\end{lemma}
\begin{proof}
	Since $v$ and all vectors in $\mathbf{S}$ can be generated by $v_{\mathbf{S}}$, we have 
	\begin{align}
		w_H(v_{\hat{\mathbf{S}}}) = w_H(v_{\mathbf{S}}) &\le \sum_{i\in \mathbf{S}} w_H(v_i) - (|\mathbf{S}|-1)d\\
		& = \sum_{i\in \mathbf{S}} (w_H(v_i)-d) + d\\
		& = |\mathbf{S}| + d < |\hat{\mathbf{S}}| +d
	\end{align}
\end{proof}
Thus, any vector $v \in \mathcal{G}(v_{\mathbf{S}},d)\setminus \mathbf{S}$ is not compatible with $\mathbf{S}$ in terms of the Constraint~\eqref{Eq:dB2} if Inequality~\eqref{Ieq:Mq1} holds.
Hence, once we find any set of basis vectors that satisfy Inequality~\eqref{Ieq:Mq}, any vector that can be generated by the merged vector should not be consider.
\begin{remark}
	Binary vector $v_m$ which has weight larger than $d+1$ can be treated as a merged vector of $\mathcal{B}(v_m,d)$.
	Therefore, $w_H(v_m)-d = |\mathcal{B}(v_m,d)|$.
	Condition~\eqref{Ieq:Mq1} also works for the cases where some of the vectors have weight larger than $d+1$.
\end{remark}
We use set $\mathbf{V}$ to store the balanced $(d,K)$-Basis vectors that have been generated by previous PDVs and set $\mathbf{Q}$ to store merged $(d,K)$-Basis vectors.
Any set of vectors which satisfy Condition~\eqref{Ieq:Mq1} will be merged as one vector and stored in $\mathbf{Q}$. 
Only $b\in \mathcal{B}(e_i,d)$ that cannot be generated by any vector in $\mathbf{Q}$ can be selected as the basis vectors.
After all vectors in $\mathcal{B}(e_i,d)$ have been checked, there must exist $e_i$ or a vector that can generate $e_i$ in $\mathbf{Q}$.

In the subspace spanned by any two vectors in $\mathbf{Q}$, there must exist at least one vector that should be added to form the $(d,K)$-Basis.
Instead of checking every subset of $\mathbf{Q}$ for merging, it is sufficient to only check the newly added vector with any subset $\mathbf{S}\subseteq \mathbf{Q}$ with $|\mathbf{S}|\le 2$ and treat the merged vector as the newly added vector for further merging until no merging possibility.

At the end, if $K-d$ such vectors are found, the PDVs of nodes are able to generate a $(d,K)$-Basis which is stored by $\mathbf{V}$ and the algorithm returns \textbf{True} and the corresponding basis $\mathbf{V}$. Otherwise, return \textbf{False}.

\begin{algorithm}[!tbp]
	\caption{Search balanced $(d,K)$-Basis (SdB)}
	\label{alg_sdb}
	\begin{algorithmic}[1]
	\State \textbf{Input:} $E = [e_1,\dots, e_N] ^\mathsf{T}$ and $d$. 
	\State \textbf{Output:} \textbf{True}, $\mathbf{r}$, $\mathbf{V}$ or \textbf{False}.
	\State \textbf{Initialization}: $\mathbf{Q} = \emptyset$, $\mathbf{V} = \emptyset$, $\mathbf{r} = [r_1,\dots,r_N]^\mathsf{T}=\mathbf{0}_{1\times N}$.
	\For{$i: i\in \{1,\dots,N\}$}
	\For{$b\in \mathcal{B}(e_i,d)$}		
	\If{$b\not \in \mathcal{G}(\mathbf{Q},d)$ } 
	\State $r_i = r_i+1$
	\State $\mathbf{V} = \mathbf{V} \cup \{b\}$
	
	\While{$ \exists \mathbf{S}\subseteq \mathbf{Q}, |\mathbf{S}|\le2: (\ref{Ieq:Mqa})\textit{ holds}$}
	\begin{equation}
	w_H(q_{\mathbf{S}}\lor b) \le \sum_{q_i\in \mathbf{S}} w_H(q_i) +w_H(b) - |\mathbf{S}|d\label{Ieq:Mqa}
	\end{equation}
	\State $b = b\lor q_{\mathbf{S}}$, $\mathbf{Q} = \mathbf{Q}\setminus \mathbf{S}$ 
	\EndWhile
	\State  $\mathbf{Q} = \mathbf{Q} \cup \{b\}$	
	\EndIf 	
	\If{$|\mathbf{V}| = K-d$}
	\State return \textbf{True}, $\mathbf{r}$ and $\mathbf{V}$
	\EndIf
	\EndFor
	\EndFor
	\State return \textbf{False}
\end{algorithmic}
\end{algorithm}

\subsection{Searching for $d^*$}
We propose Algorithm~\ref{alg_r} which uses binary search method to find the $(d^*,K)$-Basis that can be generated by PDVs of nodes.
Let $e^*$ be the PDV of the node which has largest number of available packets initially, $e^* = \arg\max_{e_i} w_H(e_i)$.
According to Theorem~\ref{Thm:CCDE}, if the PDVs of nodes can generate any $(d,K)$-Basis such that $d\ge \mathcal{M}$, we do not have to check for any larger $d$.
Also, the $(d,K)$-Basis with largest $d$ that can be generated should alway be no larger than $w_H(e^*)-1$.
Therefore, we start from $d_{max} = \min \{\mathcal{M},w_H(e^*)-1\}$ instead of $K$.
	
	\begin{algorithm}[!tbp]
		\caption{Minimal Number of Required Transmissions and d-Basis}
		\label{alg_r}
		\begin{algorithmic}[1]
			\State \textbf{Input:} $E_{N\times K} = [e_1,\dots, e_N] ^\mathsf{T}$. 
			\State \textbf{Output:} $R^*$, $\mathbf{V}^*$
			\State \textbf{Initialization}: $d_{min} = 1$, $d_{max} = \min \{\mathcal{M},w_H(e^*)-1\}$.
			\State $(F,\mathbf{r},\mathbf{V}) = \textit{SdB}(E,d_{max})$
			\If {$F$ is \textbf{True}}
			\State $d^*= d_{max}$, $\mathbf{V}^*=\mathbf{V}$
			\Else
			\State $(F,\mathbf{r},\mathbf{V}) = \textit{SdB}(E,d_{min})$
			\If{$\lnot F$}
			\State $d^*= 0$, $\mathbf{V}^* = I_{K}$			
			\Else
			\While{$d_{max}-d_{min} > 1$}
			\State $d= \lfloor \frac{d_{min}+d_{max}}{2}\rfloor$
			\State $(F,\mathbf{r},\mathbf{V}) = \textit{SdB}(E,d)$
			\If{$F$}
			\State $d_{min}=d$, $\mathbf{V}^*=\mathbf{V}$
			\Else
			\State $d_{max}=d$
			\EndIf
			\EndWhile
			\State $d^*= d_{min}$
			\EndIf
			\EndIf
			\State $R^* = K- d*$
		\end{algorithmic}
	\end{algorithm}
	
	\subsection{Complexity}
	
	In Algorithm~\ref{alg_r}, binary search method is used to find the $(d^*,K)$-Basis that can be generated by the PDVs of nodes which has complexity bounded by $\log(K)$. 
	For each specific $d$, Algorithm~\ref{alg_sdb} is used to search the existence of $(d,K)$-Basis.
	Let $M(d)$ denote the number of nodes that have at least $d+1$ packets.
	The first \textit{For} loop has at most $M(d)$ iterations.
	For the $i^{th}$ candidate PDV $e_i$, the size of set $\mathcal{B}(e_i,d)$ satisfies $|\mathcal{B}(e_i,d)|=w_H(e_i)-d$.
	Hence, the second \textit{For} loop has at most $w_H(e_i)-d$ iterations.
	The number of subsets of vectors in $\mathbf{Q}$ with size 1  and 2 are  $|\mathbf{Q}|$ and $|\mathbf{Q}| \choose 2$, respectively.
	For the $i^{th}$ checked node, $|\mathbf{Q}| \le i$, because basis vectors generated by the same PDV can always be merged to one vector and basis vectors generated by different PDV may still be merged.
	The number of possible merging iteration for each candidate basis vector is less than the size of $(d,K)$-Basis vector which is $K-d$.
	Then, the \textit{While} loop has at most $i+ {i \choose 2}(K-d)$ iterations for the $i^{th}$ PDV.
	Hence the complexity\footnote{Computing bitwise \textbf{AND} or \textbf{OR} of two $K$-dimensional binary vector has complexity of $K$ basic operations. In step 6, we compute bitwise \textbf{OR} between $b$ and each vector in $\mathbf{Q}$ and this results are also used in merging checking. Hence complexity of step 6 is not considered.} of Algorithm~\ref{alg_sdb} is bounded by $\sum_{i=1}^{M(d)} (i+ {i \choose 2})(w_H(e_i)-d)(K-d)K$.
	Since $M(d) \le N$ and $w_H(e_i) \le K$, 
	we have the overall complexity is bounded by $\mathcal{O}(N^3K^3\log (K))$, which is much lower than the complexity of existing algorithms proposed in~\cite{milosavljevic2016efficient} based on minimizing a submodular function $\mathcal{O}((N^6K^3+N^7)\log(K))$ and algorithm based on subgradient methods $\mathcal{O}((N^4\log(N)+N^4K^3)K^2\log(K))$.


	\begin{example*}\label{EX:CCDE2}
		Apply our algorithms on Example~\ref{EX:CCDE}.
		Node 4  and node 1 initially have the smallest and the largest number of packets respectively, which means $\mathcal{M} = 4$ and $w_H(e^*) = 6$.
		Therefore we have $d_{max}= 4$.
		Algorithm~\ref{alg_r} will first check whether it is possible to generate a $(4,9)$-Basis from $\{e_1,e_2,e_3\}$ by $SdB(E,4)$.
		The PDV of the $4^{th}$ node, $e_4$, will not be considered as the candidate, since $w_H(e_4)=4$ and it can not generate any binary vector with $5$ ones. 
		In this example, $SdB(E,4)$ returns \textbf{True}. 
		The minimal number of required transmissions is $5$.
		For general cases, if $(d_{max},K)$-Basis cannot be generated, binary search method would be used to find the $d^*$.
		
		Now we investigate the detail of $Sdb(E,4)$.
		The first \textit{For} loop only runs for $\{e_1,e_2,e_3\}$.
		\item 
		\begin{itemize}
			\item 	For $e_1$, $\mathcal{B}(e_1,4) = \{b_{11},b_{12}\}$ where $b_{11} = [1,1,1,1,1,0,0,0,0]$ and $b_{12}=[1,1,1,1,0,1,0,0,0]$.
			\item 	For $e_2$, $\mathcal{B}(e_2,4) = \{b_{21},b_{22}\}$ where $b_{21} = [1,1,1,0,0,0,1,1,0]$ and $b_{22}=[1,1,1,0,0,0,1,0,1]$.
			\item 	For $e_3$, $\mathcal{B}(e_3,4) = \{b_{31},b_{32}\}$ where $b_{31} = [0,0,0,1,1,1,1,1,0]$ and $b_{32}=[0,0,0,1,1,1,1,1,1]$.
		\end{itemize}
		The second \textit{For} loop runs for each $b_{ij}\in \mathcal{B}(e_i,4)$.
		\begin{itemize}
			\item For $b_{11}$, since $\mathbf{Q} = \emptyset$, $b_{11}$ will be added into $\mathbf{V}$ as $v_1$ and $\mathbf{Q}$ as $q_1$ directly.
			\item For $b_{12}$, since it cannot be generated by $q_1$, $b_{12}$ will be added into $\mathbf{V}$ as $v_2$. 
			Now, $\mathbf{Q}$ is not empty and has $q_1$. We have to check whether $b_{12}$ should be merged with $q_1$ as one vector or not.
			Since $w_H(b_{12}\lor q_1) \le w_H(b_{12}) + w_H(q_1) - d$ satisfies Inequality~\eqref{Ieq:Mqa}. We should merge them and update as $q_1 = [1,1,1,1,1,1,0,0,0]$ \footnote{In fact, $b_{12}$ and $q_1$ can be merged without checking Condition~\eqref{Ieq:Mq}, since $q_1 =b_{11}$ and $b_{12}$ are generated by the same PDV, $e_1$.}.
			\item For $b_{21}$, since it cannot be generated by $q_1$, $b_{21}$ will be added into $\mathbf{V}$ as $v_3$.
			The merging possibility between $b_{21}$ and $q_1$ will be checked and it turns out that they should not be merged. Hence $b_{21}$ will be added into $\mathbf{Q}$ as $q_2$.
			\item For $b_{22}$, since it cannot be generated by $q_1$ or $q_2$, $b_{22}$ will be added into $\mathbf{V}$ as $v_4$.
			It can be verified that $b_{22}$ should be merged with $q_2$ but not with $q_1$. Hence $q_2$ is updated as $q_2 = [1,1,1,0,0,0,1,1,1]$.
			\item For $b_{31}$, since it cannot be generated by $q_1$ or $q_2$, $b_{31}$ will be added into  $\mathbf{V}$ as $v_5$.
			Now we have enough $(4,9)$-Basis vectors. The algorithm $SdB(E,4)$ will return \textbf{True} and corresponding $\mathbf{V}$ shown as  Equation~(\ref{eq:V}).
		\end{itemize}
	    Actually, if we check the merging possibility between $b_{31}$ and $\{q_1,q_2\}$, we will find that $b_{31}$ should not be merged with $q_1$ or $q_2$ individually, but should be merged with them together. And when we have the complete $(d,K)$-Basis, we can always merge all vectors in $\mathbf{Q}$ into one vector with $K$ ones.
	    Although $b_{32}$ is in $\mathcal{B}(e_3,4)$, it is not used, because we got enough basis vectors before its iteration.
		
	\end{example*}

	\section{Code Construction}\label{sec:codeconstruction}
	In previous sections, we presented how to compute the minimum number of required transmissions and the corresponding algorithms. 
	To completely solve the CDE problem, we still need to give the coding scheme which achieves universal recovery by using the minimum number of transmissions.
	In this section, we explain how to explicitly design the optimal coding scheme.
	
	After knowing the number of transmissions which should be made by each node, 
	designing the coding scheme can be formulated as a multicast network code construction problem. 
	Methods based on the mixed matrix completion algorithm~\cite{harvey2005deterministic} and the Jaggi {\it et al.} algorithm~\cite{jaggi2005polynomial} are presented in~\cite{milosavljevic2016efficient}.
	However, those methods have to take all packet distribution information into consideration and generate a coding scheme that may only works for this particular setting. 
	As Theorem~\ref{Thm:dt} pointed out, 
	it is possible to construct a coding scheme which enables universal recovery at all nodes with at least $K-R^*$ packets.
	Packet distribution information of nodes which do not send anything is not necessary for constructing the code and can be ignored.
	This class of codes is based on MDS codes which can be constructed efficiently by starting from Vandermonde matrices.  
	
	Consider an $R\times K$ Vandermonde matrix over a finite field $\mathbb{F}_q$, where $R=K-d$:
	\begin{equation}\label{eq:vm}
	\mathcal{V} =
	\begin{bmatrix}
	1 & 1 & 1  & \dots &1&1 \\
	\theta_1 & \theta_2 & \theta_3&\dots & \theta_{K-1}&\theta_{K} \\
	\vdots & \vdots & \vdots &\ddots & \vdots &\vdots \\
	\theta_1^{R-1} & \theta_2^{R-1} & \theta_3^{R-1} &\dots & \theta_{K-1}^{R-1} &\theta_{K}^{R-1}
	\end{bmatrix}
	\end{equation}
	For large enough $q$, there exists $\{\theta_1,\dots,\theta_K\}$ such that any $m$ ($m\le R$) columns of $\mathcal{V}$ are linearly independent. 
	Apparently $\mathcal{V}$ is the generator matrix of an MDS code. 
	However, the coefficient matrix $A$ cannot simply be set equal to $\mathcal{V}$, 
	since the number of non-zero entries of each row cannot be larger than the number of available packets at the node which generates this transmission.
	Nevertheless, by performing elementary transformations on $\mathcal{V},$ we can transform it into a coefficient matrix $A$
	with the property that each row has $K-R+1$ non-zero entries. 
	
	\begin{lemma}
		For any $R\times K$ $(R\le K)$ Vandermonde matrix $\mathcal{V}$, by performing elementary row operations on $\mathcal{V}$, it is possible to get a matrix $A$ with row vectors $\{\alpha_1,\dots,\alpha_R\}$ such that
		\begin{align}
			&w_H(\alpha_i) = K-R+1 &\forall i \in [K-R]\label{eq:vd1}\\
			&w_H(\alpha_{\mathbf{S}}) \ge |\mathbf{S}| + K-R & \emptyset \ne \mathbf{S} \subseteq \{\alpha_1,\dots,\alpha_R\}\label{eq:vd2}
		\end{align}
	\end{lemma}
	\begin{proof}
		Suppose we have a $R\times K$ Vandermonde matrix depicted as Eqn.~\eqref{eq:vm}.
		We use $\mathcal{V}_l$ and $\mathcal{V}_r$ to denote the first $R$ columns submatrix and the last $K-R$ columns submatrix of $\mathcal{V}$, respectively.
		Then, $\mathcal{V} = \begin{bmatrix}
		\mathcal{V}_l & \mathcal{V}_r
		\end{bmatrix}$.
		Since any $R$ columns of $\mathcal{V}$ are linearly independent, $\mathcal{V}_l$ is always a full rank matrix and invertible.
		Performing elementary row operations on $\mathcal{V}$ is equivalent to left multiplying a $R\times R$ matrix to $\mathcal{V}$.
		Let $D$ denote a $R \times R$ matrix and $D = \mathcal{V}_l^{-1}$.
		\begin{align}
			D \mathcal{V} = D
			\begin{bmatrix}
				\mathcal{V}_l & \mathcal{V}_r
			\end{bmatrix}
			=
			\begin{bmatrix}
			I_{R} & D\mathcal{V}_r
			\end{bmatrix} 
		\end{align}
		where $I_{R}$ is the $R\times R$ identity matrix.
		Since $D$ is invertible and is a full rank matrix, we have 
		\begin{align}
			\text{rank} (D\mathcal{V}_r) = \text{rank} (\mathcal{V}_r) = K-R
		\end{align}
		Hence, $D\mathcal{V}_r$ is a column full rank matrix and each row could have $K-R$ non-zero entries.
		Let $A = D\mathcal{V}$, then we have row vectors of $A$ satisfy
		\begin{align}
			&w_H(\alpha_i) = 1+K-R &\forall i \in [K-R]\\
			&w_H(\alpha_{\mathbf{S}}) \ge |\mathbf{S}| + K-R & \emptyset \ne \mathbf{S} \subseteq \{\alpha_1,\dots,\alpha_R\}
		\end{align}
	\end{proof}

	The matrix $D\mathcal{V}$ with $D = \mathcal{V}_l^{-1}$ satisfies both conditions of the balanced $(d,K)$-Basis with $d = K-R$. Hence it can be a coefficient matrix for the coding scheme based on the $(d,K)$-Basis.
	Normally, the places of non-zero entries of the $(d,K)$-Basis generated by the PDVs of nodes are different from matrix $D\mathcal{V}$.
	However, since any row vector with $d+1$ ones is in the space spanned by row vectors of $D\mathcal{V}$, further elementary row operations can be performed on $D\mathcal{V}$ to get the coefficient matrix with non-zeros entries at the same places as the $(d,K)$-Basis generated by the PDVs of nodes.

	
	\begin{example*}
		Now we show how to use a Vandermonde matrix to construct the linear coding scheme for Example~\ref{EX:CCDE}.
		We know that $R^*=5$ and there exists a $(4,9)$-Basis $V$.
		Consider the Vandermonde matrix $\mathcal{V}$ over the finite field $GF(2^4)$ with primitive polynomial $\alpha^4+\alpha+1$.
		\begin{equation}
		\mathcal{V} =
		\begin{bmatrix}
		1 & 1 & 1 & 1 & 1 &1 &1 &1  &1 \\
		1 & 2 & 3 & 4 &5&6&7&8 & 9 \\
		1 & 4 & 5 &3&2&7&6&12 & 13 \\
		1 & 8 & 15 &12&10&1&1&10 & 15 \\
		1 & 3& 2  & 5 & 4 &6&7&15&14 
		\end{bmatrix}
		\end{equation}
		By elementary row transformations and Gaussian eliminations, we can get the coefficient matrix $A$ shown as~\eqref{eq:A}.
		Given any four packets, the other packets can be recovered from transmissions based on $A$.
		Suppose there is another node with PDV $e_5 = [1,0,1,0,1,0,0,1,0]$.
		It can also recover all its missing packets by receiving transmissions based on $A$. 
		The detail of its packet distribution information is not used for either computing the minimal number of required transmissions or designing the coding scheme.
		Although in this example the coefficient matrix of our method looks more complicated than that of methods based on Jaggi {\it et al.}'s algorithm, in general cases, the complexity of constructing the coefficient matrix via our method is much lower.
	\end{example*}


\section{Weighted Cost of Transmissions}\label{sec:weightcost}

In the basic CDE problem, every transmission incurs the same cost, irrespective of the transmitting node.
However, in more general cases, it is intuitive to consider that the (transmit) cost for different nodes are different.
Let $\mathbf{w} = [w_1,\dots,w_N]^\mathsf{T}$ denote the weight vector where each $w_i$ is the cost for node $i$ to make one transmission.
For any coding scheme with rate vector $\mathbf{r} = [r_1,\dots,r_N]^\mathsf{T}$, the weighted cost is denoted by $\mathcal{C}(\mathbf{r}) = \mathbf{w}^\mathsf{T}\cdot \mathbf{r}=\sum_{i=1}^{N} w_ir_i$.
Instead of minimizing the total number of transmissions (sum rate), the goal of the CDE problem with weighted cost is to achieve universal recovery by a coding scheme with rate vector which has  the minimum weighted cost.
Note that once the optimal rate vector is found, a corresponding optimal transmission scheme can be developed exactly along the lines of the unweighted case discussed earlier.
 
The minimum weighted cost for the CDE problem with weighted cost can be computed as
\begin{align}
\mathcal{C}^* = \min_{\mathbf{r} \in \Omega} \mathcal{C}(\mathbf{r}) =  \min_{\mathbf{r} \in \Omega} \sum_{i=1}^{N} w_ir_i .
\end{align}
Although the optimization should be over all vectors in $\Omega$, we can actually decompose this optimization problem into two sub-optimization problems.
We first find the optimal rate vector under the conditional that the sum rate is fixed.
Then, the further optimization should only be over the optimal rate vectors for different fixed sum rates. 
\begin{definition}
	Let $\mathcal{K}(R)$ denote the minimum weighted cost of all rate vectors that can achieve universal recovery and has sum-rate equal to $R$.
	\begin{align}
	\mathcal{K}(R) = \min_{\mathbf{r}\in\Omega, \mathcal{S}(\mathbf{r})=R}  \mathcal{C}(\mathbf{r}) =  \min_{\mathbf{r}\in\Omega, \mathcal{S}(\mathbf{r})=R}\sum_{i=1}^{N} w_ir_i\label{opt:fixR}					
	\end{align}
\end{definition}
Let $R_{min}$ denote the minimum sum rate such that rate vector can achieve universal recovery\footnote{In previous sections, for basic CDE problem, we use $R^*$ to denote the minimum sum rate such that universal recovery can be achieved. However, in CDE problem with weighted cost, the optimal rate vector may not have minimum sum rate.}. 
Only rate vectors with sum rate between $R_{min}$ and $K$ should be considered.
The minimum weighted cost can also be computed as
\begin{align}
\mathcal{C}^* &= \min_{R \in \{R_{min},\dots,K\}} \mathcal{K}(R)\nonumber\\
&= \min_{R \in \{R_{min},\dots,K\}} \min_{\mathbf{r}\in\Omega, \mathcal{S}(\mathbf{r})=R} \sum_{i=1}^{N} w_ir_i\label{opt:gopt}		
\end{align}

\begin{example}
	\label{EX:CCDEweight}
	Consider a CDE problem for the fully connected network with 5 nodes and 9 packets with the goal of minimizing the weighted cost of transmissions. The packet distribution matrix (PDM) is as following:
	\begin{equation*}
	E =\begin{bmatrix}
	0 & 1 & 0 & 1 & 0 & 0 & 1 & 1 & 1\\
	1 & 0 & 0 & 0 & 1 & 1 & 0 & 1 & 1\\
	0 & 1 & 1 & 0 & 0 & 1 & 0 & 1 & 1\\ 
	1 & 0 & 1 & 0 & 1 & 1 & 0 & 1 & 0\\
	1 & 1 & 0 & 1 & 1 & 0 & 1 & 0 & 1
	\end{bmatrix}
	\end{equation*}
	The weights of nodes are as following:
	\begin{center}
		\begin{tabular}{ |c|c|c|c|c|c|c|c|c| } 
			\hline
			Node(i) &  1 &  2   & 3  &  4  & 5 \\
			\hline
			\hline
			$w_i$ &  2 &  3   & 6  &  8  & 10\\
			\hline
		\end{tabular}
	\end{center}
	
	By using the methods proposed in~\cite{milosavljevic2016efficient,5743607}, 
	we can find that the optimal rate vector is $\mathbf{r}^* = [3,3,1,0,0]^\mathsf{T}$ and the minimum weighted cost is $21$.
	However, for the basic CDE problem (unweighted case) with the same packet distribution matrix, the optimal rate vector is $\mathbf{r} = [1,1,1,1,1]^\mathsf{T}$.
\end{example}

\begin{remark}
	By using algorithm in~\cite{courtade2014coded,milosavljevic2016efficient,Li1706:Cooperative}, we can show that the minimum sum rate $R_{min}$ for Example~\ref{EX:CCDEweight} is $5$. But for CDE problem with weighted cost, the optimal rate vector has sum rate $7$, which is larger than the minimum required sum rate. 
	Thus, only finding the rate vector with sum rate $R_{min}$ is not enough, we have to optimize $\mathcal{K}(R)$ over all $R \in \{R_{min},\dots,K\}$.
	However, we show that it is not necessary to compute $\mathcal{K}(R)$ for all $R \in \{R_{min},\dots,K\}$.
	By exploiting the convexity of the function $\mathcal{K}(R)$, we can search the optimal $R$ and rate vector by the binary searching method.
\end{remark}

We propose an efficient deterministic algorithm based on $(d,K)$-Basis to solve the optimization problem~\eqref{opt:fixR}. 
For a given fixed number of transmission $R$, Algorithm~\ref{alg_sdbw} searches the existence of corresponding $(d,K)$-Basis where $d =K - R$.

\begin{algorithm}[!htbp]
	\caption{Search $(d,K)$-Basis (SdB)}
	\label{alg_sdbw}
	\begin{algorithmic}[1]
		\State \textbf{Input:} $E = [e_1,\dots, e_N] ^\mathsf{T}$ ($w_i\le w_j$ $\forall i\le j$) and $d$. 
		\State \textbf{Output:} \textbf{True}, $\mathbf{r}$, $\mathbf{V}$ or \textbf{False}.
		\State \textbf{Initialization}: $\mathbf{Q} = \emptyset$, $\mathbf{V} = \emptyset$, $\mathbf{r} = [r_1,\dots,r_N]^\mathsf{T}=\mathbf{0}_{1\times N}$.
		\For{$i: i\in \{1,\dots,N\}$}
		\For{$b\in \mathcal{B}(e_i,d)$}		
		\If{$b\not \in \mathcal{G}(\mathbf{Q},d)$ } 
		\State $r_i = r_i+1$
		\State $\mathbf{V} = \mathbf{V} \cup \{b\}$
		
		\While{$ \exists \mathbf{S}\subseteq \mathbf{Q}, |\mathbf{S}|\le2: (\ref{Ieq:Mq})\textit{ holds}$}
		\begin{equation}
		w_H(q_{\mathbf{S}}\lor b) \le \sum_{q_i\in \mathbf{S}} w_H(q_i) +w_H(b) - |\mathbf{S}|d\label{Ieq:Mq}
		\end{equation}
		\State $b = b\lor q_{\mathbf{S}}$, $\mathbf{Q} = \mathbf{Q}\setminus \mathbf{S}$ 
		\EndWhile
		\State  $\mathbf{Q} = \mathbf{Q} \cup \{b\}$	
		\EndIf 	
		\If{$|\mathbf{V}| = K-d$}
		\State return \textbf{True}, $\mathbf{r}$ and $\mathbf{V}$
		\EndIf
		\EndFor
		\EndFor
		\State return \textbf{False}
	\end{algorithmic}
\end{algorithm}

\begin{theorem}\label{thm:wtopt}
	For any $R \in \{R_{min},\dots,K\}$ and $d= K-R$, let $\mathbf{r} = [r_1,\dots,r_N]^\mathsf{T}$ be the output rate vector of Algorithm~\ref{alg_sdbw} with input $E$ and $d$, then $\mathcal{K}(R) = \sum_{i=1}^{N}w_ir_i$.
\end{theorem}

The details of the proof of Theorem~\ref{thm:wtopt} are given in Appendix~\ref{Ap:proof1}, but for a brief outline, we may observe that for any other rate vector which has the same sum rate as the rate vector $\mathbf{r}$ output by Algorithm~\ref{alg_sdbw}, we must have either
(1) if it can achieve universal recovery, it has equal or larger weighted cost than $\mathbf{r}$; or
(2) it cannot achieve universal recovery, hence it should not be considered.

\begin{remark}
	In words, Theorem~\ref{thm:wtopt} says that the output rate vector of Algorithm~\ref{alg_sdbw} is the optimal rate vector  which  has the minimum weighted cost among all the rate vectors which have sum rate $R$ and can achieve universal recovery.
\end{remark}

Comparing to Algorithm~\ref{alg_sdb} which checks the existence of a balanced $(d,K)$-Basis for the basic CDE problem and outputs the corresponding $(d,K)$-Basis vectors if they exist,
Algorithm~\ref{alg_sdbw} requires that the input PDVs be ordered according to their weights.
The nodes with smaller weights have smaller indices.
The node with the smallest weight would be selected to generate as many $(d,K)$-Basis vectors as it can. 
Then, the nodes with larger weights would be selected to generate $(d,K)$-Basis vectors that can not be generated by previous nodes.
We show that by ordering the input PDVs in ascending order of their weights, Algorithm~\ref{alg_sdbw} can find the optimal rate vector and corresponding $(d,K)$-Basis vectors which can achieve universal recovery by using $K-d$ transmissions and has minimum overall weighted cost. 
The ordering of the PDVs according to their weights can be done before the start of Algorithm~\ref{alg_sdbw} and only requires complexity $\mathcal{O}(\log(N))$. As compared to the complexity of searching the existence of a $(d,K)$-Basis, which is $\mathcal{O}(N^3K^3)$, the complexity of pre-ordering nodes can be ignored.

Now we have a method to get the optimal solution to the sub-optimization problem~\eqref{opt:fixR}.
In order to get the globally optimal solution to optimization problem~\eqref{opt:gopt}, it is sufficient to only consider the rate vectors that are output by Algorithm~\ref{alg_sdbw} with different values of input parameter $d$ ($d= K-R$). 
However, it is not necessary to run Algorithm~\ref{alg_sdbw} with all possible $R \in \{R_{min},\dots,K\}$, by leveraging convexity of the function $\mathcal{K}(R)$ which is stated by the following Theorem and hence the optimal weighted cost and rate vector can be found by a binary search style method.

\begin{theorem}\label{thm:convexity}
	For $ R_{min} \le R \le K$, the function defined by~\eqref{opt:fixR}:  
	$\mathcal{K}(R) = \min_{\mathbf{r}\in\Omega, \mathcal{S}(\mathbf{r})=R}  \sum_{i=1}^{N} w_ir_i$ is convex.
\end{theorem}
The proof is given in Appendix~\ref{Ap:proof2}.
To prove Theorem~\ref{thm:convexity}, it is sufficient to only consider coding schemes with rate vectors output by Algorithm~\ref{alg_sdbw}, since they are the conditionally optimal solution for fixed sum rate $R$.
In particular, we exploit some properties of rate vector output by Algorithm~\ref{alg_sdbw} to show that the second order difference of $\mathcal{K}(R)$ is non-negative, i.e. $\mathcal{K}(R+2)+\mathcal{K}(R)-2\mathcal{K}(R+1) \ge 0$. By induction, we prove that $\mathcal{K}(R)$ is a convex function of $R$. 

\begin{remark}
	In~\cite{milosavljevic2016efficient}, it has been proved that the function $\mathcal{K}(R)$ defined in~\eqref{opt:fixR} is convex for $R_{min} \le R\le K$ for a relaxed condition where each entry of $\mathbf{r}= [r_1,\dots,r_N]^\mathsf{T}$ can be non-integer rate vector.
	However, the rate vector should always be integer for the cooperative data exchange problem.
	The improvement of our theorem is we prove that for integer rate vectors, the function $\mathcal{K}(R)$ defined in~\eqref{opt:fixR} is still convex for $R_{min} \le R\le K$.
\end{remark}

Since the function $\mathcal{K}(R)$ is a convex function, it is not necessary to search all possible $R$ to get the optimal solution to optimization problem~\eqref{opt:gopt}.
We propose Algorithm~\ref{alg_bs} to compute the minimum weighted cost by using a binary searching style method.
\begin{algorithm}[!htbp]
	\caption{Finding $\mathbf{r}^*$ and $\mathcal{C}^*$ using Binary Search Algorithm}
	\label{alg_bs}
	\begin{algorithmic}[1]
		\State \textbf{Input:} $E = [e_1,\dots, e_N] ^\mathsf{T}$, $K$ and $\mathbf{w} = [w_1,\dots, w_N] ^\mathsf{T}$ such that ($w_i\le w_j$ $\forall i\le j$)
		\State \textbf{Output:} $\mathbf{r^*}$ and $\mathcal{C}^*$
		\State \textbf{Initialization}: $d_{start} = 0$, $d_{end} = \mathcal{M}$
		
		\While {$ d_{start} <  d_{end} $}
		\State $d = \max \{\lfloor \frac{d_{start}+d_{end}}{2}\rfloor,d_{start}+1 \}$
		\State $(F,\mathbf{r}, \mathbf{V}) = SdB( E, d)$
		\If{$F$ is \textbf{False}}
		\State $d_{end} = d$
		\Else
		\State $\hat{d} = d-1$
		\State $(\hat{F},\hat{\mathbf{r}}, \hat{\mathbf{V}}) = SdB( E, \hat{d})$
		\If{$\mathbf{w}^\mathsf{T}\cdot\mathbf{r} > \mathbf{w}^\mathsf{T}\cdot\hat{\mathbf{r}}$}
		\State $d_{end} = \hat{d}$, $\mathbf{r}^* = \hat{\mathbf{r}}$
		\Else
		\State $d_{start} = d$, $\mathbf{r}^* = \mathbf{r}$
		\EndIf
		\EndIf
		\EndWhile
		\State $R^* = K-d$, $\mathcal{C}^* = \mathbf{w}^\mathsf{T}\cdot \mathbf{r}$
		
	\end{algorithmic}
\end{algorithm}

The complexity of the binary search of Algorithm~\ref{alg_bs} is approximately $\mathcal{O}(\log(K))$.
Hence, the overall complexity of our two algorithms is $\mathcal{O}(N^3K^3\log(K))$ which is the same as complexity as the complexity of algorithms for basic CDE problem. 

\begin{example3*}
	On applying Algorithm~\ref{alg_sdbw} on Example~\ref{EX:CCDEweight} for $d = \{0,1,2,3,4\}$, we can get the results as shown in Table~\ref{tb:1}. 
	\begin{table}[h]
		\centering
		\caption{Sum rate, optimal weighted cost and rate vector}\label{tb:1}
		\begin{tabular}{ |c|c|c|c|c|c|c|c| } 
			\hline
			d&R=K-d &  $\mathcal{K}(R)$ &  $r_1$ &$ r_2 $  &$ r_3 $& $r_4$  & $r_5$ \\
			\hline
			\hline
			4&5 & 29 & 1 & 1 & 1& 1& 1  \\ 
			\hline
			3&6 & 22 & 2 & 2 & 2& 0& 0  \\  
			\hline
			2&7 & 21 & 3 & 3 & 1& 0& 0 \\ 
			\hline
			1&8 & 23 & 4 & 3 & 1& 0& 0 \\ 
			\hline
			0&9 & 25 & 5 & 3 & 1& 0& 0  \\ 
			\hline
		\end{tabular}
	\end{table}
	As can be seen from the table, the minimum cost is achieved by a coding scheme that uses $7$ transmissions, which is larger than the minimum number of required transmissions ($R_{min}=5$) for achieving universal recovery. Additionally, if we plot the function $\mathcal{K}(R)$ vs $R$ for example~\ref{EX:CCDEweight} and connect the points, it is easy to see the convexity in Fig.~\ref{fig:1}.
	
	\begin{figure}[!htbp]
		\centering
		\includegraphics[width=0.3\columnwidth]{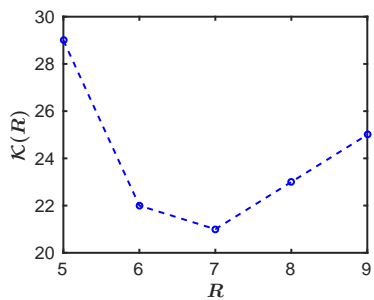}
		\caption{Optimal weighted cost ($\mathcal{K}(R)$) vs Sum rate ($R$) for Example~\ref{EX:CCDEweight}.}
		\label{fig:1}
	\end{figure}	
\end{example3*}

\section{Successive Local Omniscience}\label{sec:slo}
In the basic CDE problems, all nodes have the same priority and should be able to recover all packets at the end of the communication phase.
In this section, we consider a generalized problem called {\it Successive Local Omniscience (SLO)}\cite{8006942} where nodes have different priorities. Specifically, let $\mathbf{G} = \{\mathbf{G}_1,\dots,\mathbf{G}_M\}$ be a partition of node set $\mathbf{N}.$ In the SLO problem, communication occurs in $M$ rounds, numbered from 1 to $M$ and taking place in this order, as follows:
\begin{list}{$\bullet$}{\setlength{\itemsep}{-1pt}\setlength{\topsep}{-1pt}}
\item In round $i,$ only the nodes in the set $\mathbf{G}_{[i]} \stackrel{\mathrm{def}}{=} \cup_{j=1}^{i} \mathbf{G}_j$ are allowed to transmit.
\item After round $i,$ all nodes in the set $\mathbf{G}_{[i]}$ must be able to recover all packets that were initially present at all the nodes in the set $\mathbf{G}_{[i]}.$
\end{list}
In this sense, if $i<k,$ then nodes in $G_i$ can be thought of as having priority over nodes in $G_k$ (although in the general case, no node is guaranteed to attain full omniscience of {\it all }  packets before the end of the last round).

Let $\mathbf{r}^{i} = [r_1^i,\dots,r_N^i]^{\mathsf{T}}$ denote the accumulated rate vector up to and including the $i^{th}$ round, where each $r_j^i$ denotes the total number of transmissions made by node $j$ from the first round to the $i^{th}$ round.
The corresponding entries of rate vectors $\mathbf{r}^i$ and $\mathbf{r}^{i+1}$ satisfy $r_j^i \le r_j^{i+1}$ for every node $j \in \mathbf{N}.$
Let $\Omega (\mathbf{G}_{[i]})$ be the set of rate vectors up to and including the $i^{th}$ communication round satisfying
\begin{align}
	\sum_{j\in \mathbf{G}_{[i]} \setminus \mathbf{I}}r_j^i \ge \left| \mathbf{X}_{\mathbf{G}_{[i]}} \setminus \mathbf{X}_\mathbf{I}\right|, \forall \mathbf{I}\subsetneq \mathbf{G}_{[i]}
\end{align}
Then we have the following lemma characterizing solutions to the SLO problem:
\begin{lemma}
	Any solution to the SLO problem is also a solution to the following multi-objective linear program:
	\begin{align}
	\min_{\mathbf{r}^i \in \Omega(\mathbf{G}_{[i]}) } \sum_{j=1}^{N} r_j^i, \forall i \in [M]
	\end{align}
	
\end{lemma}
\begin{proof}
	For any $i\in [M]$, rate vectors $\mathbf{r}^i \in \Omega(\mathbf{G}_{[i]})$ satisfy the Slepian-Wolf constraints for
	achieving local omniscience and only nodes in $\mathbf{G}_{[i]}$ are allowed to make transmissions. 
	The minimization gives the minimum sum rate. 
	Thus, for $M$ communication rounds, the overall optimal solutions achieve successive local omniscience.
\end{proof}

The main goal and contribution of this section is to present a more efficient solution of the SLO problem via the $(d,K)$-Basis method.
Let $E_{\mathbf{G}_{[i]}}$ denote the packet distribution matrix of the nodes in $\mathbf{G}_{[i]}$. 
If we run Algorithm~\ref{alg_r} with $E_{\mathbf{G}_{[i]}}$ as input in the subspace indexed by the collectively available packets of $\mathbf{G}_{[i]}$, it will return the minimum number of required transmissions for achieving local omniscience as well as the corresponding $(d,K_i)$-Basis vectors.
Algorithm~\ref{alg_r} can be called for every $E_{\mathbf{G}_{[i]}}$, $i \in [M]$ and we can get the $d_i$-Basis vectors for local omniscience achieved by each $\mathbf{G}_{[i]}$.
If $d_i \ge d_{i+1}$, the $d_i$-Basis vectors can also be used to generate $d_{i+1}$-Basis vectors by adding $0$'s to the dimensions that are added by packets in $\mathbf{X}_{\mathbf{G}_{[i+1]}}\setminus \mathbf{X}_{\mathbf{G}_{[i]}}$.
If $d_i < d_{i+1}$, the $d_i$-Basis vectors cannot be used to generate $d_{i+1}$-Basis vectors.
Hence, the optimal strategy is to use the coding scheme based on $d_{i}$-Basis in the subspace indexed by packets of $\mathbf{X}_{\mathbf{G}_{[i+1]}}$ so that every transmissions used in the previous round are useful in the current round.

\begin{theorem}
	For successive local omniscience problem with $\mathbf{G}_{[i]}$ and corresponding packet distribution submatrix $E_{\mathbf{G}_{[i]}}$, for $i \in [M]$, the minimum number of required transmissions $R_i^*$ for round $i$ is
	\begin{align}
		R_i^* = K_i -\min\{\mathcal{M}_i, d_1^*,\dots,d_i^*\}
	\end{align} 
	where $K_i = |\mathbf{X}_{\mathbf{G}_{[i]}}|$ is the number of packets collectively available at nodes in $\mathbf{G}_{[i]}$, $\mathcal{M}_i = \min_{j \in \mathbf{G}_{[i]}} |X_j|$ is the minimum number of available packets at any single node in $\mathbf{G}_{[i]}$ and $d_i^*$ is the maximum $(d,K_i)$-Basis that can be generated by PDVs of nodes in $\mathbf{G}_{[i]}$.
\end{theorem}
\begin{proof}
	For the first round, $R_1^* = K_1 - \min\{\mathcal{M}_1,d_1^*\}$, according to Theorem~\ref{Thm:CCDE}.
	For the $i^{th}$ round, since $\mathbf{G}_{[j]} \subset \mathbf{G}_{[i]}$, $\forall j < i$, $\mathcal{M}_i \le \mathcal{M}_j$.
	According to Theorem~\ref{Thm:CCDE}, nodes in $\mathbf{G}_{[i]}$ can generate a coding scheme based on $\{\mathcal{M}_i,d_i^*\}$-Basis can achieve local omniscience.
	If $d_i^* = \min \{d_1^*\dots,d_i^*\}$, and all transmissions used in previous rounds can also be used as the transmissions of coding schemes based on $\{\mathcal{M}_i,d_i^*\}$-Basis.
	Thus, in the $i^{th}$ round, only additional transmissions are required and the total minimum number of required transmissions for achieving local omniscience is  $R_i^* = K_i - \min\{\mathcal{M}_i,d_i^* \}$.
	If $d_j^* = \min \{d_1^*\dots,d_i^*\}$ and $j < i$, then transmissions generated in the $j^{th}$ round cannot all be used for coding scheme based on $\{\mathcal{M}_i,d_i^*\}$-Basis.
	In order to make use of all previously generated transmissions, coding scheme based on $\{\mathcal{M}_i,d_j^*\}$-Basis can be used to achieve local omniscience for nodes in $\mathbf{G}_{[i]}$ and the total number of required transmissions is $R_i^* = K_i - \{\mathcal{M}_i,d_j^*\}$.
	Therefore, $R_i^* = K_i -\min\{\mathcal{M}_i, d_1^*,\dots,d_i^*\}$.
\end{proof}

We propose Algorithm~\ref{alg_slo} to compute the minimum number of required transmissions ($R_i^*$) and the local optimal rate vector ($\mathbf{r}_i^*$) for nodes in each group with different priorities. Algorithm~\ref{alg_slo} iteratively calls Algorithm~\ref{alg_sdb} to find the existence of $(d,K_i)$-Basis that can be generated for linear coding scheme to achieve local omniscience.

\begin{algorithm}[!htbp]
	\caption{Successive Local Omniscience}
	\label{alg_slo}
	\begin{algorithmic}[1]
		\State \textbf{Input:} $E = [e_1,\dots, e_N] ^\mathsf{T}$ and $\mathbf{G} = \{\mathbf{G}_1,\dots,\mathbf{G}_M\}$
		\State \textbf{Output:} $R_1^*,\dots,R_M^*$ and $\mathbf{r}_1^*,\dots,\mathbf{r}_M^*$
		\State \textbf{Initialization}: $d^* = K$		
		\For{$i = 1 \dots M$}
		\State $d_{min} = 1$, $d_{max} = \min \{\mathcal{M}_i,d^*\}$
		\State $(F,\mathbf{r},\mathbf{V}) = SdB(E_{\mathbf{G}_{[i]}},d_{end})$
		\If{$F$ is \textbf{True}}
			\State $d_i^* = d_{max}$, $\mathbf{V}_i^* = \mathbf{V}$, $\mathbf{r}_i^* = \mathbf{r}$
		\Else
			\State $(F,\mathbf{r},\mathbf{V}) = SdB(E_{\mathbf{G}_{[i]}},d_{min})$
			\If{$F$ is \textbf{False}}
				\State $d_i^* = 0$, $\mathbf{V}_i^* = I_{K_i}$
			\Else
				\While{$d_{max} - d_{min} > 1$}
				\State $d= \lfloor \frac{d_{min}+d_{max}}{2}\rfloor$
				\State $(F,\mathbf{r},\mathbf{V}) = \textit{SdB}(E,d)$
				\If{$F$ is \textbf{True}}
				\State $d_{min}=d$, $d_i^* = d$, $\mathbf{V}_i^*=\mathbf{V}$, $\mathbf{r}_i^* = \mathbf{r}$
				\Else
				\State $d_{max}=d$
				\EndIf
				\EndWhile					
			\EndIf

		\EndIf
		
		\State $d^* = d_i^*$, $R_i^* = K_i - d_i^*$
		\EndFor
		
	\end{algorithmic}
\end{algorithm}

Based on the $(d,K_i)$-Basis vectors $\mathbf{V}_i^*$ and local optimal rate vector $r_i^*$, the corresponding linear coding scheme can be generated to achieve local omniscience. Instead of generating linear coding scheme for each communication round individually, it is possible to globally generate a linear coding scheme in which the first $R_i^*$ transmissions can achieve local omniscience.

In terms of the complexity of our approach, in each communication round, the minimum number of required transmissions and the accumulated rate vector are found by using binary search method and iteratively call Algorithm~\ref{alg_sdb}. 
The total number of outer iteration is equal to the number of priority groups, $M$.
The binary search method for the $i^{th}$ round has complexity bounded by $\mathcal{O}(\log(K_i))$.
For the $i^{th}$ round, Algorithm~\ref{alg_sdb} has complexity bounded by $\mathcal{O}(|\mathbf{G}_{[i]}|^3K_i^3)$, since the number of nodes and packets considered in the $i^{th}$ round are $|\mathbf{G}_{[i]}|$ and $K_i$, respectively.
Hence, the total number of computation can be expressed as $\sum_{i=1}^{M}|\mathbf{G}_{[i]}|^3K_i^3\log(K_i)$.
Since $\forall i:$ $|\mathbf{G}_{[i]}|\le N$ and $K_i \le K$, the overall complexity of our $(d,K)$-Basis method for solving SLO problem is bounded by $\mathcal{O}(N^3K^3M\log(K))$.

\begin{example}
	Consider the successive local omniscience problem with packet distribution matrix 
	\begin{align}
		E =
		\begin{bmatrix}
		1 & 1 & 1 & 1 & 0 & 0 & 0 & 0 & 0	\\
		0 & 1 & 1 & 1 & 1 & 0 & 0 & 0 & 0	\\
		1 & 1 &  0& 0 & 0 & 1 & 0 & 0 & 0	\\
		0 & 0 & 1 & 1 & 0 & 0 & 1 & 0 & 0	\\
		1 & 0 & 1 & 1 & 1 & 1 & 1 & 1 & 0	\\
		1 & 1 & 1 & 1 & 1 & 0 & 1 & 0 & 1	
		\end{bmatrix}
	\end{align}
	And the nodes are partitioned into three groups with decreasing priorities: $\mathbf{G}_1 = \{1,2\}$, $\mathbf{G}_2 = \{3,4\}$ and $\mathbf{G}_3 = \{5,6\}$.
	Since nodes in $\mathbf{G}_1$ collectively only have packets $P_1,\dots,P_5$, the optimization for the first communication round is equivalent to the basic CDE problem with packet distribution matrix $E_{\mathbf{G}_1}$, which is a submatrix of the first two rows of $E$.
	\begin{align}
		E_{\mathbf{G}_1} =
		\begin{bmatrix}
		1 & 1 & 1 & 1 & 0 	\\
		0 & 1 & 1 & 1 & 1 	\\
		\end{bmatrix}
	\end{align}
	It is apparent that only two transmissions are required to achieve local omniscience for $\mathbf{G}_{1}$.
	Consider the following two coding schemes:
	\begin{itemize}
		\item Coding scheme 1:  Node 1 sends $P_1$ and Node 2 sends $P_5$.
		\item Coding scheme 2: Node 2 sends $P_1+P_2+P_3+P_4$ and Node 2 sends $P_2+P_3+P_4+P_5$.
	\end{itemize}
	In Coding scheme 1, each transmission is a linear combination of as few packets as possible, while in Coding scheme 2, each transmission is a linear combination of as many packets as possible.
	Both coding schemes can enable two nodes to fully recover packets that are collectively available at them.
	However, we will show that Coding scheme 1 is suboptimal but Coding scheme 2 is optimal.
	In the second communication round, the goal is to enable node in $\mathbf{G}_{[2]}$ to recover packets which are collectively available at them.
	Similarly, we have packet distribution matrix $E_{\mathbb{G}_{[2]}}$, which is a submatrix of the first four rows of $E$.
	\begin{align}
		E_{\mathbf{G}_{[2]}} =
		\begin{bmatrix}
		1 & 1 & 1 & 1 & 0 & 0 & 0	\\
		0 & 1 & 1 & 1 & 1 & 0 & 0 \\
		1 & 1 &  0& 0 & 0 & 1 & 0 	\\
		0 & 0 & 1 & 1 & 0 & 0 & 1 	
		\end{bmatrix}
	\end{align}
	If we treat this as a packet distribution matrix of a basic CDE problem,  it is easy to find that the minimum number of required transmission is $5$, since the $(2,7)$-Basis is the $(d,7)$-Basis with largest $d$ value that can be generated by row vectors of $E_{\mathbf{G}_{[2]}}$.
	And this implies that in the successive local omniscience problem, the total number of required transmissions is at least $5$.
	If we choose Coding scheme 1 in the first transmission round, the packet distribution matrix becomes
	\begin{align}
	\hat{E}_{\mathbf{G}_{[2]}} =
	\begin{bmatrix}
	1 & 1 & 1 & 1 & 1 & 0 & 0	\\
	1 & 1 & 1 & 1 & 1 & 0 & 0 \\
	1 & 1 &  0& 0 & 1 & 1 & 0 	\\
	1 & 0 & 1 & 1 & 1 & 0 & 1 	
	\end{bmatrix}
	\end{align}
	As the row vectors of $\hat{E}_{\mathbf{G}_{[2]}}$ can only generate a $(3,7)$-Basis which has largest $d$ value, $4$ transmissions are required in the second communication round to achieve local omniscience for nodes in $\mathbf{G}_{[2]}$.
	Hence, the total number of transmissions for the first and second rounds is $2+4=6$ which is larger than the lower bound $5$.
	However, if coding scheme 2 is chosen in the first round,
	actually it is possible to generate a coding scheme based on $(2,7)$-Basis in which the first two transmissions achieve local omniscience for nodes in $\mathbf{G}_1$.
	The desired $(2,7)$-Basis generated by $E_{\mathbf{G}_{[2]}}$ is 
	\begin{align}
		\begin{bmatrix}
		v_1\\v_2\\v_3\\v_4\\v_5
		\end{bmatrix}
		=
		\begin{bmatrix}
		1 & 1 & 1 & 1 & 0 & 0 & 0	\\
		0 & 1 & 1 & 1 & 1 & 0 & 0	\\
		0 & 1 & 1 & 1 & 1 & 0 & 0	\\
		1 & 1 & 0 & 0 & 0 & 1 & 0	\\
		0 & 0 & 1 & 1 & 0 & 0 & 1	\\
		\end{bmatrix}
	\end{align}
	As you can see the first $5$ columns of $v_1$ and $v_2$ can actually form a $(3,5)$-Basis.
	And the coding scheme based on them can achieve local omniscience for nodes in $\mathbf{G}_1$.
	Similarly, we can show that $2$ transmissions are required in the third communication round to achieve omniscience for nodes in $\mathbf{G}_{[3]}$.
	Instead of generating coefficients for linear combinations of packets for each round individually, we can deal with them together by constructing a linear coding scheme based on the final $(d,K)$-Basis we need, which is $(2,9)$-Basis in this case.
	Given the rate vector in each round:
	\begin{align}
		\mathbf{r}_1 = [1,1,0,0,0,0]^{\mathsf{T}}\\
		\mathbf{r}_2 = [0,1,1,1,0,0]^{\mathsf{T}}\\
		\mathbf{r}_3 = [0,0,0,0,1,1]^{\mathsf{T}}
	\end{align}
	And the $(2,9)$-Basis that generated by row vectors of $E$
		\begin{align}
	\begin{bmatrix}
	v_1\\v_2\\v_3\\v_4\\v_5\\v_6\\v_7
	\end{bmatrix}
	=
	\begin{bmatrix}
	1 & 1 & 1 & 1 & 0 & 0 & 0 & 0 & 0	\\
	0 & 1 & 1 & 1 & 1 & 0 & 0 & 0 & 0	\\
	0 & 1 & 1 & 1 & 1 & 0 & 0 & 0 & 0	\\
	1 & 1 & 0 & 0 & 0 & 1 & 0 & 0 & 0	\\
	0 & 0 & 1 & 1 & 0 & 0 & 1 & 0 & 0	\\
	1 & 0 & 1 & 1 & 1 & 1 & 1 & 1 & 0	\\
	1 & 1 & 1 & 1 & 1 & 0 & 1 & 0 & 1	
	\end{bmatrix}
	\end{align}
	By using the coding construction method based on MDS code in Section~\ref{sec:codeconstruction}, we can get a coefficient matrix as follows, where all entries are over finite file $GF(2^4)$ with primitive polynomial $\alpha^4+\alpha +1$.
	\begin{align}
		\begin{bmatrix}
			a_1\\a_2\\a_3\\a_4\\a_5\\a_6\\a_7
		\end{bmatrix}
		=
		\begin{bmatrix}
		4 & 7 & 3 & 1 & 0 & 0 & 0 & 0 & 0	\\
		0 & 8 & 12 & 3 & 2 & 0 & 0 & 0 & 0	\\
		0 & 13 & 13 & 2 & 2 & 0 & 0 & 0 & 0	\\
		15 & 8 & 0 & 0 & 0 & 1 & 0 & 0 & 0	\\
		0 & 0 & 4 & 5 & 0 & 0 & 10 & 0 & 0	\\
		10 & 0 & 9 & 9 & 5 & 5 & 5 & 5 & 0	\\
		9 & 4 & 1 & 1 & 1 & 0 & 1 & 0 & 1	
		\end{bmatrix}
	\end{align}
	It can be verified that the first $2$ transmissions achieves local omniscience for nodes in $\mathbf{G}_1$, the first $5$ transmissions achieve local omniscience for nodes in $\mathbf{G}_{[2]}$, and all transmissions together achieve omniscience for nodes in $\mathbf{G}_{[3]}$ (all nodes).
\end{example}

\newpage

\section{Conclusion}\label{sec:conclusion}
In this paper, we introduce the notion of the $(d,K)$-Basis. We establish that the existence of such a basis is both a necessary and sufficient condition for the existence of coding schemes that can achieve universal recovery with $K-d$ transmissions for the fully connected network. 
We provide a polynomial-time deterministic algorithm based on the $(d,K)$-basis construction which solves the cooperative data exchange problem.
We show that we can efficiently construct the coefficients of an optimal linear coding scheme starting from a Vandermonde matrix by levering the connection between the $(d,K)$-Basis and maximum distance separable codes.
Moreover, we demonstrate that our $(d,K)$-Basis construction method can also be used in solving generalized versions of the cooperative data exchange problem, including with weighted cost and with successive local omniscience. 

\section*{Acknowledgment}
This work was supported in part by the Swiss National Science Foundation under Grant 169294. The authors thank Abhin Shah for his contribution for Section~\ref{sec:weightcost} during his summer internship at EPFL.

\appendix
\subsection{Proof of Theorem~\ref{thm:wtopt}}\label{Ap:proof1}
In order to prove Theorem~\ref{thm:wtopt}, we first prove two useful Lemmas.
\begin{lemma}\label{lm:w2}
	Let $\mathbf{r}^* = [r_1^*,r_2^*,\dots,r_N^*]^\mathsf{T}$ denote the rate vector output by Algorithm~\ref{alg_sdbw}.
	For any rate vector $\mathbf{r} = [r_1,\dots,r_N]^\mathsf{T}$ such that $\mathbf{r} \in \Omega$ and $\mathcal{S}(\mathbf{r}^*) = \mathcal{S}(\mathbf{r}) $, there does not exists any node pair $(i,j)$ such that $i< j$, $r_i > r_i^*$ and $r_j < r_j^*$.
\end{lemma}

\begin{proof}
	If the coding scheme with rate vector $\mathbf{r}$ can achieve universal recovery and uses the same total number of transmissions, then the coding scheme can be implemented as a $(d,K)$-Basis based coding scheme which has the same $d$ value as the coding scheme with rate vector $\mathbf{r}^*$.
	As Algorithm~\ref{alg_sdbw} guarantees that $\forall i \in [N]$, if $r_i^* > 0$, then there must exist as many as $\sum_{j=i}^{N} r_j^*$ $(d,K)$-Basis vectors that cannot be generated by nodes in set $\{1,2,\dots,i-1\}$. If $\exists i< j$ such that, $r_i > r_i^*$ and $r_j < r_j^*$, then $\sum_{j=i}^{N} r_j < \sum_{j=i}^{N} r_j^*$ which is not possible as such vectors can only be generated by nodes in set $\{i,i+1\dots,N\}$.
	Hence, it is impossible that $\exists i< j$: $r_i > r_i^*$ and $r_j < r_j^*$.
\end{proof}

\begin{lemma}\label{lm:w3}
	Let $\mathbf{r}^* = [r_1^*,r_2^*,\dots,r_N^*]^\mathsf{T}$ denote the rate vector output by Algorithm~\ref{alg_sdbw}. If there exists a coding scheme with rate vector $\mathbf{r} = [r_1,\dots,r_N]^\mathsf{T}$ such that $\mathbf{r} \in \Omega$, $\mathcal{S}(\mathbf{r}^*) = \mathcal{S}(\mathbf{r})$. If there exists node pair (i,j) such that $i< j$, $r_i < r_i^*$ and $r_j > r_j^*$, then $\mathcal{C}(\mathbf{r}) \geq \mathcal{C}(\mathbf{r}^*)$.
\end{lemma}
\begin{proof}
	Let $\mathbf{S}_1 = \{i: r_i < r_i^*\}$, $\mathbf{S}_2 = \{j: r_j > r_j^*\}$ and $\mathbf{S}_3 = \{k: r_k = r_k^*\}$.
	Since  $\mathcal{S}(\mathbf{r}^*) = \mathcal{S}(\mathbf{r}) $ , we have
	\begin{align}
	0=\sum_{i = 0}^{N} (r_i -r_i^*) &= \sum_{i \in \mathbf{S}_1} (r_i - r_i^*)+\sum_{j \in \mathbf{S}_2} (r_j - r_j^*)+\sum_{k \in \mathbf{S}_3} (r_k - r_k^*)\\
	&=\sum_{i \in \mathbf{S}_1} (r_i - r_i^*)+\sum_{j \in \mathbf{S}_2} (r_j - r_j^*)
	\end{align}
	Hence, for each $i\in \mathbf{S}_1$ that sends one less transmission, there must exist one corresponding $j\in \mathbf{S}_2$ which sends one more transmission.
	According to Lemma~\ref{lm:w2}, if there exists such pair of $(i,j)$, it must satisfy $i< j$ and $w_i < w_j$.
	Let $P = \sum_{i \in \mathbf{S}_1} (r_i^* - r_i)= \sum_{j \in \mathbf{S}_2} (r_j - r_j^*)$ denote the total number of such pairs and $\mathcal{P}$ denote the partition of such pairs. 
	Therefore,
	\begin{align}
	\mathcal{C}(\mathbf{r}) - \mathcal{C}(\mathbf{r}^*) & = \sum_{i = 0}^{N} w_ir_i -  \sum_{i = 0}^{N} w_ir_i^* \\
	&= \sum_{i \in \mathbf{S}_1} w_i(r_i - r_i^*) +\sum_{j \in \mathbf{S}_2} w_j(r_j - r_j^*)\\
	&= \sum_{(i,j) \in  \mathcal{P}} (w_j - w_i )\\
	&\ge 0
	\end{align}
\end{proof}
Now, we are ready to prove Theorem~\ref{thm:wtopt}.
\begin{proof}[\textbf{Proof of Theorem~\ref{thm:wtopt}}]
	If there exists any linear coding scheme that achieves universal recovery by using $K-d$ transmissions with rate vector $\mathbf{r}=[r_1,\dots,r_N]^\mathsf{T}$ ($\sum_{i=i}^{N} r_i = K-d$), it is always possible to generate a corresponding linear coding scheme based on $(d,K)$-Basis that have the same rate vectors~\cite{Li1706:Cooperative}.
	Hence, they have the same weighted cost and we can only consider the coding schemes based on $(d,K)$-Basis. 
	Let $r^* = [r_1^*,\dots,r_N^*]^\mathsf{T}$ denote the rate vector output by Algorithm~\ref{alg_sdbw}.
	According to Lemma~\ref{lm:w2}, there does not exist any $i<j$ such that $r_j<r_j^*$. 
	Additionally, since $\mathcal{S}(\mathbf{r}^*) = \mathcal{S}(\mathbf{r}) $,
	if rate vector $\mathbf{r}$ is different from $\mathbf{r}^*$, the change can only be $\exists i<j:$ $r_i<r_i^*$ and $r_j>r_j^*$.	
	According to Lemma~\ref{lm:w3}, $\mathcal{C}(\mathbf{r}) \geq \mathcal{C}(\mathbf{r}^*) $.
	Therefore, the rate vector output by Algorithm~\ref{alg_sdbw}  has minimum weighted cost in all coding schemes which use $K-d$ transmissions and achieve universal recovery.
\end{proof}

\subsection{Proof of Theorem~\ref{thm:convexity}}\label{Ap:proof2}
In order to prove Theorem~\ref{thm:convexity}, we first prove two useful Lemmas.
\begin{lemma}\label{lm:ws}
	Let $\mathbf{r}(l)$ be the rate vector output by Algorithm~\ref{alg_sdbw} for input $E$ and $d = K- l$. Thus, $\mathbf{r}(l)$ is the optimal rate vector with minimum weighted cost among all the rate vectors with $\mathcal{S}(\mathbf{r}) = l$.
	For the coding schemes with rate vectors $\mathbf{r}(l) = [{r}_{(l,1)},\dots, {r}_{(l,N)}]^\mathsf{T}$ with $l \in \{R_{min},\dots,K\}$ yielded by Algorithm~\ref{alg_sdbw},we have
	\begin{itemize}
		\item[(1)] ${r}_{(l+1,1)} = {r}_{(l,1)} +1$.
		\item[(2)] ${r}_{(l+1,m)} \le {r}_{(l,m)} +1$, $\forall 2\le m \le N$.
		\item[(3)] If ${r}_{(l+1,m)} < {r}_{(l,m)}$, then ${r}_{(l+2,m)} \leq {r}_{(l+1,m)}$.
	\end{itemize}
\end{lemma}

\begin{proof}
	(1) Since in Algorithm~\ref{alg_sdbw}, we always start the generation of basis vectors from the PDV of node 1 and there is no previously generated basis vector, then the number of basis vectors that should be generated by node 1 is 
	\begin{align}
	{r}_{(l,1)} = w_H(e_1)- d = w_H(e_1) - K+ l
	\end{align}
	Since $R_{min}\le l \le K$ and $R_{min} = K- \min \{\mathcal{M},d^*\}$, we have $0 \le {r}_{(l,1)} \le w_H(e_1)$.
	Note that $w_H(e_1) \ge \mathcal{M} \ge K-R_{min}$.
	Therefore, for any feasible $l$, we have ${r}_{(l+1,1)} = {r}_{(l,1)} +1$.
	This means the first node generates 1 more vector when the total number of transmissions increases by 1. 
	When ${r}_{(l,1)} = |X_1| $, each transmissions is just a pure packet.
	In such cases, we have $d=0$ and $l=K$.
	Universal recovery can always be achieved when all packets have been sent individually.
	No coding scheme with more than $K$ transmissions should be considered.
	
	(2) Similarly, for any $2\le  m \le N$, the total number of feasible basis vectors that can be generated by node $m$ is $w_H(e_m) - K + l$.
	However, some of them may not be compatible with basis vectors that have been generated by previous nodes.
	Hence we have
	\begin{align}
	r_{(l,m)} \le w_H(e_m) - K + l
	\end{align}
	And ${r}_{(l+1,m)} \le {r}_{(l,m)} +1$, $\forall 2\le m \le N$.
	This means node $m$ can generate at most 1 more basis vector when the total number of transmissions increases by 1.
	
	(3) As the total number of transmissions (sum rate) goes from $l$ to $l+1$, the corresponding basis changes from $(K-l)$-Basis to $(K-l-1)$-Basis.
	Therefore, the number of packets that are used to generate each transmission decreases by 1. 
	Note that $w_H(e_m) \ge \mathcal{M} \ge K-R_{min}$ , $\forall m \in [N]$. When $l = R_{min}$, nodes $m$ with $w_H(e_m) = K-R_{min}$  are not considered to generate any basis vector, since every basis vector needs $K-R_{min}+1$ ones.  
	But when $l > R_{min}$, every node is considered to generate basis vectors.
	If node $i$ is not used to generate any basis vector, that means all basis vectors that can be generated by node $i$ are not compatible with the basis vectors generated by previous nodes.
	If ${r}_{(l+1,m)} < {r}_{(l,m)}$, that means besides the first node, there exists at least one node with lower weight than node $m$ that generates more basis vector(s), i.e. $\exists n$ s.t. $n <m$ and $r_{(l+1,n)} > r_{(l,n)}$.
	The set of basis vectors that are generated to form $(K-l-1)$-Basis by node $m$ is a subset of $\mathcal{B}(e_m,K-l)$.
	Let $\mathcal{D}(m,l+1)$ denote vectors in $\mathcal{B}(e_m,K-l)$ but are not selected to form $(K-l-1)$-Basis.
	Then every vector in $\mathcal{D}(m,l+1)$ is not compatible with $(K-l-1)$-Basis vectors generated by priouves nodes.
	Any vector in $\mathcal{B}(e_m,K-l-1)$ which can be generated by vectors in $\mathcal{D}(m,l+1)$ is also not compatible with $(K-l-2)$-Basis vectors generated by priouves nodes.
	Hence, the maximum number of basis vectors that can be generated by node $m$ for next round is upperbounded by $r_{(l+1,m)}$.
	Therefore, If ${r}_{(l+1,m)} < {r}_{(l,m)}$, then ${r}_{(l+2,m)} \leq {r}_{(l+1,m)}$ , $\forall 2\le m \le N$.

	%
	
	
\end{proof}

\begin{definition}
	Let $\mathbf{S}_{(l,\uparrow)}$ denote the set of nodes which generate more number of transmissions when the sum rate increases from $l$ to $l+1$. Let $\mathbf{S}_{(l,0)}$ denote the set of nodes which generate the same number of transmissions when the sum rate increases from $l$ to $l+1$.
	Let $\mathbf{S}_{(l,\downarrow)}$  denote the multiset of nodes which generate fewer transmissions when the sum rate increases from $l$ to $l+1$. The multiplicity of node $i$ in $\mathbf{S}_{(l,\downarrow)}$  equals $r_{(l,i)} - r_{(l+1,i)}$.
\end{definition}

\begin{lemma}\label{lm:ws4}
	For  $ \forall  R_{min} \le l \le K-1$, we have $(1)$ $\mathbf{S}_{(l+1,\uparrow)} \subseteq \mathbf{S}_{(l,\uparrow)}$ and 
	(2) Let $W_{l+1}^i$ be the $i^{th}$ largest $w \in \{w_j:j\in \mathbf{S}_{(l+1,\downarrow)}\}$ and $W_l^i$ be the $i^{th}$ largest $w\in \{w_j: j\in \mathbf{S}_{(l,\downarrow)}\}$. For any $W_{l+1}^{i}$, there exists $W_l^{i}$ such that $W_{l+1}^i\le W_l^i$.
\end{lemma}

\begin{proof}
	Let $\mathbf{r}(l) = [r_{(l,1)},\dots,r_{(l,N)}]^\mathsf{T}$ and $\mathbf{r}(l+1) = [r_{(l,1)},\dots,r_{(l+1,N)}]^\mathsf{T}$ denote the rate vectors output by Algorithm~\ref{alg_sdbw} for $d=K-l$ and $d=K-l-1$, respectively.
	According to Theorem~\ref{thm:wtopt}, $\mathbf{r}(l)$ and $\mathbf{r}(l)$ are optimal rate vectors for fixed sum rate $l$ and $l+1$, respectively.
	
	(1) Assuming that $\mathbf{S}_{(l+1,\uparrow)}\not \subseteq \mathbf{S}_{(l,\downarrow)}$, then there must exist at least one node $k$, such that $k \in  \mathbf{S}_{(l+1,\uparrow)}$ and $k \notin  \mathbf{S}_{(l,\uparrow)}$.
	Hence, $k$ must be in $\mathbf{S}_{(l,0)}$ or $\mathbf{S}_{(l,\downarrow)}$.
	It is apparent that $k\ne 1$, since the first node always increases the rate by 1 when the total sum-rate increases by 1.
	For $ k \in  \mathbf{S}_{(l+1,\uparrow)} \setminus \{1\}$, there must always exist a corresponding node $ m \in  \mathbf{S}_{(l+1,\downarrow)}$ such that $w_k < w_m$.
	\begin{itemize}
		\item[(i)] If $ k \in \mathbf{S}_{(l,0)}$, we know that $r_{(l+1,k)} = r_{(l,k)}$. 
		Coding scheme with rate vector $\hat{\mathbf{r}}(l) = [\hat{r}_{(l,1)},\dots,\hat{r}_{(l,N)}]^\mathsf{T}$ such that
		\begin{align}
		\hat{r}_{(l,k)} &= r_{(l+1,k)} =r_{(l,k)}+1\\
		\hat{r}_{(l,m)} &= r_{(l+1,m)} = r_{(l,m)}-1\\
		\hat{r}_{(l,i)} & = r_{(l,i)}, \forall i\in [N]\setminus \{k,m\}
		\end{align}
		can also achieve universal recovery.
		Moreover, coding scheme with rate vector $\hat{\mathbf{r}}(l)$ has lower cost than coding scheme with rate vector $\mathbf{r}(l)$.
		This contradicts that coding scheme with rate vector $\mathbf{r}(l)$ is optimal for all rate vector with sum rate $l$. 
		\item[(ii)] If $ k \in  \mathbf{S}_{(l,\downarrow)}$, we know that $r_{(l+1,k)} < r_{(l,k)}$.
		According to Lemma~\ref{lm:ws}, ${r}_{(l+1,k)} \leq {r}_{(l,k)}$. 
		This contradicts our assumption that $k \in \mathbf{S}_{(l+1,\uparrow)}$.
	\end{itemize}
	Thus we have $\mathbf{S}_{(l+1,\uparrow)} \subseteq \mathbf{S}_{(l,\uparrow)}$.
	
	(2)We use induction proof method to prove this part of lemma.
	For $i = 1$, let $W_{l+1}^1 = w_m$, $W_{l}^i = w_n$. We assume that $W_{l+1}^i > W_{l}^i$, then we have $w_m > w_n$ which implies that $m \not\in \mathbf{S}_{(l,\downarrow)}$.
	Since $\mathbf{S}_{(l+1,\uparrow)} \subseteq \mathbf{S}_{(l,\uparrow)}$, coding scheme with rate vector $\hat{\mathbf{r}}(l) = [ \hat{r}_{(l,1)},\dots,\hat{r}_{(l,N)}]^\mathsf{T}$ which satisfies
	\begin{align}
	\hat{r}_{(l,m)} &=r_{(l,m)}-1\\
	\hat{r}_{(l,n)} &= r_{(l,n)}+1\\
	\hat{r}_{(l,j)} & = r_{(l,j)}, \forall j\in [N]\setminus \{m,n\}
	\end{align}
	can also achieve universal recovery with the same sum-rate and has lower weighted cost. This contradicts that coding scheme with rate vector $\mathbf{r}(l) = [r_{(l,1)},\dots,r_{(l,N)}]^\mathsf{T}$ is optimal for all rate vector with sum rate $l$. 
	Thus, we have $W_{l+1}^1 \le W_l^1$.
	For $i > 1$, assuming that $W_{l+1}^{i-1} \le W_{l}^{i-1}$, we show that  $W_{l+1}^{i} \le W_{l}^{i}$.
	let $W_{l+1}^i = w_a$, $W_{l}^i = w_b$. 
	If $W_{l+1}^{i-1} \le w_b$, then it is straightforward that $w_a = W_{l+1}^i \le W_{l+1}^{i-1}\le w_b= W_l^i$.
	If $W_{l+1}^{i-1} > w_b$, and we assume that $w_a > w_b$. In such cases, $a \not\in \mathbf{S}_{(l,\downarrow)}$, since $w_a\le W_{l+1}^{i-1} \le W_{l}^{i-1}$.
	By using similar trick as we used for $i=1$, it is able to show that there exist another coding scheme which achieves universal recovery and has lower sum weighted cost. Hence the assumption $w_a > w_b$ can never be true.
	Therefore, $W_{l+1}^i \le W_l^{i}$.
	
\end{proof}

Now we are ready to prove Theorem~\ref{thm:convexity}.

\begin{proof}[\textbf{Proof of Theorem~\ref{thm:convexity}}]
	For any $ R_{min} \le l \le K-2$, we show that the second order difference of $\mathcal{K}(l)$ is non-negative, i.e. 
	$\mathcal{F}(l+1) -\mathcal{F}(l) \ge 0$, where $\mathcal{F}(l)= \mathcal{K}(l+1)-\mathcal{K}(l)$.
	We compute the difference of the weighted cost of two coding schemes when sum-rate increases by 1. 
	\begin{align}
	\mathcal{F}(l+1) &=\mathcal{K}(l + 2) - \mathcal{K}(l+1) \\
	&= \sum_{i\in \mathcal{S}_{(l+1,\uparrow)}} w_i -  \sum_{i\in \mathcal{S}_{(l+1,\downarrow)}} w_i \\
	& = w_1 + \sum_{i\in \mathcal{S}_{(l+1,\uparrow)} \setminus \{1\}} w_i -  \sum_{i\in \mathcal{S}_{(l+1,\downarrow)}} w_i
	\end{align}
	According to Lemma~\ref{lm:ws}, node 1 always generates 1 more transmission when the total number of transmissions increases by 1. 
	And for other nodes,  if their rate increases, the increment is 1, whereas if their rate decreases, the decrement can be more than 1. And the number of multiplications of the nodes in $\mathcal{S}_{(l+1,\downarrow)}$ is equal to the decrease in rate.
	Similarly, for sum-rate change from $l$ to $l + 1$, we have
	\begin{align}
	\mathcal{F}(l) =\mathcal{K}(l+1) - \mathcal{K}(l) =w_1 + \sum_{i\in \mathcal{S}_{(l,\uparrow)} \setminus \{1\}} w_i -  \sum_{i\in \mathcal{S}_{(l,\downarrow)}} w_i
	\end{align}
	The reason why node 1 is separated from other nodes is that the total number of transmissions only increases by 1, which implies that the total number of transmissions sent by other nodes, except node 1, remains the same. 
	Hence 
	\begin{align}
	| \mathcal{S}_{(l,\uparrow)} \setminus \{1\}| &=  |\mathcal{S}_{(l,\downarrow)} |\\
	| \mathcal{S}_{(l+1,\uparrow)} \setminus \{1\}|& =  |\mathcal{S}_{(l+1,\downarrow)} |
	\end{align}
	Therefore, $\forall i\in \mathcal{S}_{(l,\uparrow)} \setminus \{1\}$, $\exists j\in \mathcal{S}_{(l,\downarrow)}$ such that $w_i < w_j$.
	We can construct a partition of node pairs $(i,j)$, where $i \in \mathcal{S}_{(l,\uparrow)}\setminus \{1\}$ and  $j \in \mathcal{S}_{(l,\downarrow)} $ as follows
	\begin{equation}
	\mathcal{P}(l) = \{(i,j):  i\in \mathcal{S}_{(l,\uparrow)} \setminus \{1\}, j\in \mathcal{S}_{(l,\downarrow)}, i<j\}
	\end{equation}
	Note that the number of node pairs in $\mathcal{P}(l)$ is equal to  $| \mathcal{S}_{(l,\uparrow)} \setminus \{1\}| $.
	Then we have
	\begin{align}
	\mathcal{F}(l) = w_1+ \sum_{(i,j)\in \mathcal{P}(l)} (w_i-w_j)		
	\end{align}
	where every term of the summation ($w_i-w_j$) is negative.
	
	We show that for each pair $(i,j)\in \mathcal{P}(l+1)$, there always exists a pair $(\hat{i},\hat{j})\in \mathcal{P}(l)$ such that 
	\begin{align}
	w_i-w_j - (w_{\hat{i}}-w_{\hat{j}}) \ge 0\label{eq:weightpair}
	\end{align}
	
	Assuming that there exists a node pair $(i,j) \in \mathcal{P}(l+1)$ such that for all possible pairs $(\hat{i},\hat{j})\in \mathcal{P}(l)$:
	\begin{align}
	w_i-w_j - (w_{\hat{i}}-w_{\hat{j}}) < 0
	\end{align}
	Equivalently, we have
	\begin{align}
	w_i-w_j < \max_{\hat{i}\in \mathcal{S}_{(l,\uparrow)}, \hat{j} \in \mathcal{S}_{(l,\downarrow)}} (w_{\hat{i}} - w_{\hat{j}})
	\end{align}
	If $i \in \mathcal{S}_{(l,\uparrow)}$, then $w_j > \max_{\hat{j}\in \mathcal{S}_{(l,\downarrow)}} w_{\hat{j}}$, which contradicts Lemma~\ref{lm:ws4}.
	If $i \not \in \mathcal{S}_{(l,\uparrow)}$, consider another coding scheme with rate vector $\mathbf{r} = [r_1,r_2,\dots,r_N]^\mathsf{T}$ such that
	\begin{align}
	r_i = r_{(l+1,i)}+1, r_j = r_{(l+1,j)}-1\\
	r_{\hat{i}} = r_{(l+1,\hat{i})}-1, r_{\hat{j}} = r_{(l+1,\hat{j})}+1\\
	r_m = r_{(l+1,m)}, \forall m \not\in \{i,j,\hat{i},\hat{j}\}
	\end{align}
	It can be verified that this coding scheme can also achieve universal recovery with total $l+1$ transmissions.
	It has lower weighted cost than the coding scheme with rate vector $[r_{(l+1,1)},\dots,r_{(l+1,N)}]^\mathsf{T}$, which contradicts that coding scheme with rate vector $[r_{(l+1,1)},\dots,r_{(l+1,N)}]^\mathsf{T}$ has the minimum weighted cost over all coding schemes that achieve universal recovery with $l+1$ transmissions.
	Starting form the node pair $(i,j)$ with largest $j$, we can apply this binding for every $(i,j)$ and remove used $(\hat{i},\hat{j})$ iteratively.
	And it is able to find $(\hat{i},\hat{j})\in \mathcal{P}(l)$ such that Eqn~\eqref{eq:weightpair} is satisfied for every pair $(i,j) \in \mathcal{P}(l+1)$.
	Hence, we have
	\begin{align}
	&\mathcal{F}(l+1)-\mathcal{F}(l)\nonumber \\
	&= \sum_{(i,j)\in \mathcal{P}(l+1)} (w_i-w_j)-\sum_{(m,n)\in \mathcal{P}(l)} (w_m-w_n)\\
	&= \sum_{(i,j)\in \mathcal{P}(l+1), (\hat{i},\hat{j}) \in \mathcal{P}(l)} [(w_i-w_j)-(w_{\hat{i}}-w_{\hat{j}})] \nonumber\\
	&- \sum_{(m,n)\in \mathcal{P}(l) \setminus \{\mathbf{Q}\}} (w_m-w_n)\\
	&\ge 0	
	\end{align}
	where every $(w_i-w_j)-(w_{\hat{i}}-w_{\hat{j}}) \ge 0$, every $w_m-w_n < 0$ and $\mathbf{Q}$ is the set of node pairs $(\hat{i},\hat{j})$ that are used in the first summation.
	Hence, the function $\mathcal{K}(l) = \min_{\mathbf{r}\in\Omega, \mathcal{S}(\mathbf{r})=l}  \sum_{i=1}^{N} w_ir_i$ is convex.
\end{proof}


\bibliographystyle{IEEEtran}
\bibliography{edic}

\end{document}